\documentclass[11pt]{article}
\usepackage[long]{alec}

\title{How to sell a service with uncertain outcomes}

\author{Krishnamurthy Iyer, Alec Sun, Haifeng Xu, You Zu}

\begin{document}

\maketitle
\begin{abstract}
    Motivated by the recent popularity of machine learning training services, we introduce a contract design problem in which a provider sells a service that results in an outcome of uncertain quality for the buyer. The seller has a set of actions that lead to different distributions over outcomes. We focus on a setting in which the seller has the ability to commit to an action and the buyer is free to accept or reject the outcome after seeing its realized quality. We propose a two-stage payment scheme where the seller designs a menu of contracts, each of which specifies an action, an upfront price and a vector of outcome-dependent usage prices. Upon selecting a contract, the buyer pays the upfront price, and after observing the realized outcome, the buyer either accepts and pays the corresponding usage price, or rejects and is exempt from further payment. We show that this two-stage payment structure is necessary to maximize profit: only upfront price or only usage prices is insufficient. We then study the computational complexity of computing a profit-maximizing menu. While computing the optimal profit is \tbf{NP}-hard even for two buyer types, we derive a fully-polynomial time approximation scheme (FPTAS) for the optimal profit for a constant number of buyer types. Finally, we prove that in the single-parameter setting in which buyers' valuations are parametrized by a single real number, there exists a \emph{revenue-optimal} menu consisting of a single contract.

\end{abstract}
\newpage
\tableofcontents
\newpage

\section{Introduction} \label{introduction}

Motivated by the surge in companies offering machine learning services, we study how to price such a service through the lens of \emph{algorithmic contract theory} \citep{dutting2024algorithmic}. Our problem setting is motivated by AutoML services offered by Google Vertex AI and Amazon SageMaker as well as enterprise large language models (LLMs) sold by \citet{openai-enterprise}. AutoML services help users automatically train their models, and the service price is based on the amount of cloud computing resources users consume. Selling fine-tuned LLMs to businesses, such as the LLM enterprise product of \citet{openai-enterprise}, has also recently become a profitable industry, with custom models costing millions of dollars. The above settings share two characteristics. First, they require the service provider to exert costly effort, whether in the form of cloud computing resources for AutoML or engineering and computing effort for fine-tuning LLMs. Second, the outcome has high uncertainty since the performance of AutoML and fine-tuned LLMs is very problem dependent and difficult to predict in advance.

To capture these characteristics, we model pricing a training service as a contract design problem between a service provider (seller) and a customer (buyer). The service consists of the seller choosing one of finitely many possible actions, each of which incurs a cost for the seller and leads to a known distribution over possible \emph{outcomes}. In our motivating applications, the action can be interpreted as the effort level that the seller undertakes to train a machine learning model. Each effort level has a different training cost for the seller. The outcome can be understood as the model quality, which is uncertain given the stochastic nature of training. The buyer is one of several possible buyer types and has a type-dependent value for every outcome. The seller seeks to maximize their \emph{profit}, which is the revenue collected from the buyer minus the expected cost of performing the service.


To sell the service, the seller presents a \emph{menu} of contracts to the buyer, where each contract has a two-part tariff structure. The two-part tariff has been widely studied in economics as a means to extract higher revenue through better price discrimination \citep{hayes1987competition, schlereth2010optimization, leland1976monopoly, armstrong2011competitive, murphy_price_1977, danaher_optimal_2002}. Formally, each \emph{contract} specifies an action, an \emph{upfront} price, and an outcome-dependent \emph{usage} price. After selecting a contract, the seller implements a two-part payment mechanism. The seller commits to the contract's action and charges the buyer an upfront price for performing the action. Upon the action's completion, the buyer is able to observe the outcome and decide whether to accept or reject it. If the buyer accepts the outcome, the seller charges the buyer a further usage price, whose value depends on the outcome. If the buyer rejects the outcome, they cannot use the product but are exempt from further payment.

We refer to the buyer's freedom to accept or reject the outcome and the seller's ability to control the buyer's access to the product as the \emph{voluntary usage} assumption. This is in contrast to \emph{mandatory usage} in which the buyer is required to purchase the product \citep{bernasconi2024agent}. A natural question is how in practice the seller can prevent the buyer from using the product if they reject the outcome. In applications of our interest, the outcome is the quality of the final model, which can be observed through performance tests. AutoML platforms like VertexAI charge buyers for model training costs (upfront price), then they deploy the trained model on their own platform and charge buyers for each call of the trained model (usage price). In LLM fine-tuning, platforms like OpenAI release the trained model as a black box and hide the model weights, so the buyer is required to keep interacting with the seller through prompts to use the model, paying per token. In both examples the service provider controls the buyer's access to the model.




\paragraph{Contributions.} Our central research question is the following: how good is our pricing scheme at maximizing seller profit compared to other models, and how hard is it to compute profit-maximizing contracts in our model? Our paper makes the following main contributions:
\begin{itemize}
    
    \item \textbf{Necessity of two-part tariff structure.} 
    We show that implementing a two-part tariff in the service provider problem yields significantly greater seller profit compared to using only upfront payments (\cref{upfront-payment-only}) or using only usage payments (\cref{usage-payment-only}), which justifies the inclusion of both types of payments in our mechanism. We provide a tight bound on the worst-case multiplicative gap in the seller profit achievable by our two-part tariff structure versus one that only uses upfront payments.
    
    \item \textbf{Superiority of voluntary usage.} We prove that maximum seller profit is \emph{always} weakly higher when the seller allows the buyer the freedom to accept or reject the outcome than when the seller forces the buyer to accept and pay for every outcome (\cref{voluntary-usage-subsumes}).
    
    \item \textbf{Complexity of computing profit-maximizing menus.}  We show that maximum profit can be achieved by offering just two distinct usage prices: 0 and $\infty$ (\cref{two-usage-prices-suffice}). Using this reduction, we prove that computing the exact maximum seller profit in the service provider problem is \tbf{NP}-hard even when there are only two buyer types and a single seller action (\cref{np-hardness-two-types}). Despite this hardness result, we use a dynamic program framework \citep{woeginger2000} to derive a fully polynomial time approximation scheme (FPTAS) for \emph{approximating} the maximum seller profit when the number of buyer types $T$ is constant and the seller's profit margin is at least a positive constant (\cref{fptas}). 
    
    \item \tbf{Revenue-optimality in single-parameter settings.} Even though the general service provider problem is \tbf{NP}-hard, when buyers' valuations are parametrized by a single real number we show that not only can we efficiently compute a \emph{revenue}-maximizing menu, but also that revenue can be maximized by a menu consisting of a single contract (\cref{single-parameter-revenue}).
\end{itemize}

Taken together, our results provide comprehensive insights into the structure of optimal contracts to sell services with uncertain outcomes.
\section{Related work} \label{connections}

In standard Bayesian contract design problems \citep{alon2023bayesian,guruganesh2021contracts, castiglioni2024reduction}, the principal issues a contract and an agent is paid to perform actions that benefit the principal. However, in our model the principal (provider) both issues contracts and performs actions that benefit the agent (customer). Our service provider problem has close connections to several problems in the mechanism design literature.


\paragraph{Selling hidden actions.}

\citet{bernasconi2024agent} study a related problem of selling a service modeled by a hidden action. Though both models are variants of principal-agent problems in which the seller performs the action, there are two fundamental differences between their setting and ours. First, we assume that the seller action is not hidden from the buyer  but rather can be committed to. We argue that this absence of \emph{moral hazard} is a natural assumption in our service provider problem. From a practical perspective, automated machine learning (AutoML) platforms such as Vertex AI and SageMaker are large-scale and backed by highly regulated parent companies and thus can commit to performing the services they offer. From a theoretical perspective, it is well-known that commitment leads to higher leader utility compared to no commitment in leader-follower games \citep{von2010leadership}, hence there is a clear economic incentive for the seller (leader) in our contract design setting to be able to commit to actions. Second, we do not require the buyer to purchase the end product and instead give the buyer the option to \emph{reject} the outcome, in which case they do not receive the product but are also not required to pay for it. We show that under this \emph{voluntary usage} assumption, seller profit is, perhaps surprisingly, always weakly higher than \emph{mandatory usage}, which is when the seller forces the buyer to accept and pay for every outcome.

\paragraph{Selling lotteries.}

Since a key feature of our model is that the service has uncertain outcomes, it is naturally related to the well-studied problem of selling lotteries, for example see \citet{chen2015complexity}. A \emph{lottery} draws an \emph{item}, which is analogous to the \emph{outcome} in our model, from a set $Q$, with different items having different probabilities of being drawn. A buyer's value for a lottery is the expected value of their valuation for the item that the lottery draws.
The fundamental difference between lottery pricing and our pricing problem is that a lottery seller has the freedom to design arbitrary lottery distributions, whereas in our model only those distributions that are achievable by provider actions are available. Nonetheless, the lottery pricing problem can be viewed as a special case of our service provider problem where there are infinitely many actions that induce all of the possible outcome distributions. The upfront payment in our model corresponds to the lottery price. We prove in \cref{usage-payment-lottery-pricing} that in lotteries, the maximum seller revenue remains the same with or without usage payments, so the two-part tariff structure in lottery pricing is not needed. In contrast, usage payments in our service provider problem are crucial, as we show in \cref{usage-payment-necessary}.

\paragraph{Selling products of differing qualities.}

\citet{mussa1978monopoly} and \citet{maskin1984monopoly} initiated the study of pricing products of differing qualities. Their model can be thought of as an instance of our service provider problem where each action deterministically maps to a unique quality and hence actions and qualities are interchangeable. Our \emph{outcomes} correspond to the \emph{qualities} in \citep{mussa1978monopoly} and our action costs are their production costs. While \citet{mussa1978monopoly} and \citet{maskin1984monopoly} assume continuous qualities and cost functions, our focus in this work is on analyzing the computational complexity of the service provider problem with discrete model primitives and where actions lead to uncertain outcomes.

\paragraph{Multidimensional screening.}

Another axis on which our work diverges from \citep{mussa1978monopoly} and \citep{maskin1984monopoly} is in the heterogeneity of consumer preferences. In \citep{mussa1978monopoly} and \citep{maskin1984monopoly}, it is natural to think of consumers' preferences as one-dimensional, ordered by the willingness to pay for higher quality products. However, in our discrete model buyer types can rank outcomes arbitrarily and valuations are not parametrized by a single number. Hence our work is more closely related to the \emph{multidimensional screening} literature \citep{rochet2003economics} in which contracts can be structured very generally based on quality and type. We show in \cref{two-usage-prices-suffice} that in our two-stage payment model, the general multidimensional screening problem reduces to a simple form where in each contract, the seller sets usage payments for some outcomes to 0 and others to $\infty$. This is equivalent to essentially blocking buyer types from purchasing certain qualities and allowing free usage for the rest after upfront payment.
\section{Model} \label{model}

We now formally describe our model for selling a service. A problem instance consists of a tuple $(A, Q, [T])$, where $A$ is a finite set of actions, $Q$ is a finite set of outcomes, and $[T]$ is a finite set of $T$ buyer types. Let $\Delta_Q= \{ \mathbf{p} \in [0,1]^Q: \sum_{q\in Q} p_q = 1 \}$ be the simplex supported on $Q$. Each action $a\in A$ incurs cost $c(a)$ for the seller and leads to a distribution $\mathbf{p}^a \in \Delta_Q$ over outcomes, with $p^a_q$ denoting the probability with which outcome $q$ is realized. This stochastic outcome distribution is introduced to capture the highly uncertain performance outcomes when providers sell training services to downstream users \citep{sun2023automl}.

The buyer type $t \in [T]$ is drawn from a distribution $\mu = (\mu^1,\lds,\mu^T) \in \Delta_{[T]}$ known to the seller. Without loss of generality, we assume that $\mu^t > 0$ for all $t\in [T]$. Each type $t$ is associated with a \emph{valuation vector} $\mathbf{v}^t \in \bR^Q$ such that $v^t_q$ is type $t$'s value for outcome $q$. As is common in contract design, we assume that the buyer type does not affect the action-to-outcome transition probabilities \citep{dutting2024algorithmic, guruganesh2021contracts, guruganesh2023menus, zuo2024new}.

\paragraph{Menus of contracts.}

We define a \emph{menu} $\cM$ as a set of contracts, where each contract $\cC = (a, w, \mathbf{x}) \in \cM$ specifies an action $a \in A$, an upfront price $w \in \bR_{\ge 0}$, and a vector of usage prices $\mathbf{x} \in \bR^Q_{\geq 0}$ where $x_q$ is the usage price for outcome $q \in Q$.

\paragraph{Seller-buyer interaction protocol.}

The interaction between the seller and buyer goes through the following steps:

\begin{enumerate}
    \item The seller commits to a menu $\cM $ of contracts and presents it to the buyer.
    
    \item A buyer with type $t$ drawn from the type distribution $\mu$ arrives and selects a contract $\cC = (a, w, \mbf{x}) \in \cM$.
    
    \item The buyer pays the seller the upfront price $w$ to enter into the contract, and the seller performs the action $a$ that they committed to.
    
    \item An outcome $q\sim \mathbf{p}^a$ is realized and observed by the buyer.
    
    \item After observing the outcome, the buyer can either choose to accept the outcome at price $x_q$ or reject it and pay nothing. 
\end{enumerate}

We refer to the assumption that a buyer is free to accept or reject the outcome as the \emph{voluntary usage} assumption. This is in contrast to the \emph{mandatory usage} assumption, where the buyer must accept the realized outcome $q$ at the price $x_q$. Under voluntary usage, a buyer of type $t$ will accept the realized outcome $q$ if and only if their value for the outcome is at least the corresponding usage price: $v_q^t \geq x_q$. Here we assume, as is common in  principal-agent problems, that the buyer breaks ties in favor of the seller. The expected utility of a buyer with type $t$ that selects a contract $\cC = (a, w, x) \in \cM$ is thus given by
$$U(t; \cC) \coloneq \sum_{q\in Q}  p^{a}_q \cdot  \max\{v^t_q - x_q, 0\} - w.$$ 

\paragraph{Incentive compatibility and direct menus.} \label{IC}

We can without loss of generality assume that the size of the menu $\cM$ is at most $T$ because each buyer type will select only the contract that yields the highest expected utility for them. By relabeling the utility-maximizing contract for type $t$ in the menu as $\cC^t$ and allowing for duplicate contracts if multiple types select the same contract, we can equivalently define a \emph{direct} menu as a tuple of contracts $\cM = \bp{\cC^t}_{t\in [T]}$, one for every type. The term \emph{direct} refers to indexing a contract by the type that selects it. This reduction from (indirect) menus to direct menus using the revelation principle is standard in contract design \citep{castiglioni2021bayesian}. 

Given a direct menu $\cM = (\cC^t)_{t \in [T]}$, we sometimes refer to $\cC^t$ as the contract $t$. For notational brevity, we sometimes denote the buyer's utility $U(t; \cC^{t'})$ for contract $\cC^{t'}$ as $U(t; t')$. We further let $U(t)$ denote the utility $U(t;t)$ a buyer of type $t$ derives from contract $\cC^t$. Under this notation, the preceding discussion implies that it suffices to consider menus that satisfy the following:

\begin{definition}[Incentive compatibility (IC)]
A menu $\cM = \bp{\cC^t}_{t\in [T]}$  is called \emph{incentive-compatible (IC)} if $U(t)  \geq U(t;t')$ for all $t,t'\in [T]$.
\end{definition}

The IC constraints imply that it is optimal for a type $t$ buyer to choose contract $\cC^t$ since no other contract $\cC^{t'}$ will yield a higher utility for them. 

\paragraph{Individual rationality.}

Finally, we assume that the buyer is always allowed to not select any contract and opt-out of the mechanism entirely. Thus, any contract selected by a buyer must yield nonnegative buyer surplus, a constraint commonly known as \emph{individual rationality} (IR).

\begin{definition}[Individual rationality (IR)]
A menu $\cM = \bp{\cC^t}_{t\in T}$ is called \emph{individually rational (IR)} if $U(t)  \geq 0$ for all buyer types $t$.
\end{definition}

\paragraph{The service provider optimization problem.}

Given an IC and IR menu $\cM = \bp{\cC^t}_{t \in [T]}$, the seller's expected profit is given by 

\begin{align*}
    \pi(\cM) \coloneq \bE_{t \sim \mu} \left[ w^t  - c(a^t) + \sum_{q\in Q} p^{a^t}_q x^t_q \cdot \one\bb{v^t_q \ge x^t_q} \right]
\end{align*}

The seller's goal is to design a menu that maximizes their expected profit, which amounts to solving the following optimization problem:

\newcommand{\Ropt}{\Pi_{\tt{opt}}}

\begin{tcolorbox}[title=Maximizing profit of a direct menu]
\vspace{-1em}
\begin{align}
    \Ropt := \max_{\cM = \bp{\cC^t}_{t \in [T]}} & \pi(\cM) \label{profit-maximization-direct-menu} \\
         U(t) &\geq U(t; t') & & \fl t, t' \in [T]  \notag \\ 
         U(t) &\geq 0  && \fl t \in [T] \notag
\end{align}
\end{tcolorbox}

Henceforth, for brevity we refer to expected profit as simply \emph{profit}.
\section{Necessity of two-part tariffs and voluntary usage} \label{two-payments}

The seller-buyer interaction protocol described in \cref{model} is more complex than the protocols for selling hidden actions \citep{bernasconi2024agent} and selling lotteries \citep{chen2015complexity} because we implement a \emph{two-part tariff}: a pricing scheme with two payment stages. For the information structure in our model in which the seller performs an action and then reveals the outcome's quality, the two-part tariff is the most general contract form because there are only two distinct information sets for the buyer: the first where they know nothing about the quality beyond the prior, and the second where they know the quality exactly.


While the two-part tariff is a widely studied mechanism \citep{hayes1987competition}, it is unclear whether two-stage payment is {\em necessary} to achieve optimal seller profit in our setting. To elaborate, while it is trivial that the two-stage payment achieves \emph{weakly} higher seller profit than a single stage alone, it is unclear whether one can \emph{strictly} improve seller profit. It is in this sense that we use the term \emph{necessary}. We begin this section by showing that indeed both stages of payment are necessary to maximize profit by quantifying the multiplicative loss associated with single-stage payments (\cref{upfront-payment-only} and \cref{usage-payment-only}). This is in contrast with both the setting with the mandatory usage assumption (\cref{voluntary-usage-subsumes}) as well as the closely related lottery selling problem (\cref{usage-payment-lottery-pricing} in \cref{a:usage-payment-lottery-pricing}), where we prove that there is no loss in restricting to single-stage payments. Finally, we leverage \cref{upfront-payment-only} to show that the seller profit under voluntary usage is always higher than the profit under mandatory usage (\cref{voluntary-usage-subsumes}).

\subsection{Usage payments are necessary} \label{usage-payment-necessary}

\newcommand{\Rupfront}{\Pi_{\tt{upfront}}}

To prove that usage payments are needed to achieve maximum seller profit, we consider a restricted version of the service provider problem where all usage prices are set to 0. Recall that $\Ropt$ denotes the maximum seller profit achievable through the two-part tariff. Let $\Rupfront$ denote the maximum profit achievable through upfront payments only. To state our results, we first make the following definition:
\begin{definition}
   For a buyer type distribution $\mu \in \Delta_{[T]}$, define  
    \begin{equation}\label{eq:H-mu}
    H_\mu \coloneq \sum_{t\in [T]} \fr{\mu^{\sigma(t)}}{\sum_{i=1}^t \mu^{\sigma(i)}}
    \end{equation}
    where $\sigma : [T] \to [T]$ is a permutation satisfying $0 < \mu^{\sigma(1)} \le \cds \leq \mu^{\sigma(T)} \leq 1$. 
\end{definition}

Observe that if $\mu\in \Delta_{[T]}$ is the uniform buyer type distribution in which $\mu^t = \fr{1}{T},\fl t\in [T]$, then $H_\mu = \sum_{t\in [T]} \fr{1}{t} = H_T$, the $T$-th harmonic number. In \cref{a:H_mu} we show that $H_\mu\in [H_T, T)$ for all $\mu\in \Delta_{[T]}$. With this definition in place, we have the following result:
\begin{proposition} \label{upfront-payment-only}
     Let $\mu \in \Delta_{[T]}$ denote the buyer type distribution. Then we have the following:
    \begin{enumerate}
        \item There exists a problem instance for which $\fr{\Ropt}{\Rupfront} = H_\mu$.
        \item For all problem instances, $\fr{\Ropt}{\Rupfront} \le H_\mu$.
    \end{enumerate}
\end{proposition}

\begin{proof}

\begin{enumerate}
    \item Since $H_\mu$ is invariant under permutations of buyer types, we can assume without loss of generality that $0<\mu^1\le \cds \le \mu^T \le 1$. We construct the following problem instance:
    \begin{itemize}
        \item Let $Q = [T]$. There is a single action $a$ with cost $c(a) = 0$ and transition probabilities $p^a_q = \fr{1}{T},\fl q\in Q$.
        \item Valuations are given by $v^t_q = \case{\fr{T}{\sum_{i=1}^t \mu^i} & \teif q = t \\ 0 & \teif q \neq t.}$
    \end{itemize}

    The utility of buyer type $t$ satisfies $$U(t;\cC^t) \le \fr{1}{T} \cd \fr{T}{\sum_{i=1}^t \mu^i} = \fr{1}{\sum_{i=1}^t \mu^i}.$$ The IR condition implies that the seller profit is upper bounded by the buyer utility, so $$\Pi_{\tt{opt}} \le \sum_{t\in [T]} \mu^t \cd U(t;\cC^t) = H_\mu.$$ This upper bound on $\Ropt$ can be achieved using an identical contract $$\cC = \bp{a, 0, \bp{\fr{T}{\sum_{i=1}^t \mu^i}}_{t\in [T]}}$$ for all types, so $$\Pi_{\tt{opt}} = \sum_{t\in [T]} \fr{\mu^t}{\sum_{i=1}^t \mu^i} = H_\mu.$$
        
    On the other hand, if the seller can only use upfront payments, note that any IC menu with a single action consists of a single contract since the contract with the lowest upfront payment yields the highest utility for all buyers. Hence all buyers will choose the contract with the lowest upfront payment. We perform casework on the value of the upfront payment $w$ by partitioning the space of possible upfront payments into disjoint ranges. In the range $w\in \left(\fr{1}{\sum_{i=1}^{t+1} \mu^i}, \fr{1}{\sum_{i=1}^{t} \mu^i}\right]$, exactly the first $t$ buyer types will choose the contract over the opt-out option. The seller revenue in this range is upper bounded by $\sum_{i=1}^t \mu^i \cd \fr{1}{\sum_{i=1}^{t} \mu^i} = 1$, with equality if $w = \fr{1}{\sum_{i=1}^{t} \mu^t}$ for some $t$. We conclude that $$\Rupfront = 1\implies \fr{\Pi_{\tt{opt}}}{\Rupfront} = \sum_{t\in [T]} \fr{\mu^t}{\sum_{i=1}^t \mu^i}.$$

    \item The contract $\cC^t = (a^t, w^t, \mathbf{x}^t)$ yields profit $$\pi^t := w^t - c(a^t) + \sum_{q\in Q} p^{a^t}_q x^t_q \cdot \one\bb{v^t_q \ge x^t_q}$$ from type $t$. By reordering the buyer types we can assume without loss of generality that $\pi^1 \ge \pi^2 \ge \cds \ge \pi^T$. We claim that for every $t\in [T]$, we can construct a modified menu with only upfront prices that achieves a profit of at least $\bp{\sum_{i=1}^t \mu^i} \cd \pi^t$. To prove this claim, replace every contract $\cC^u = (a^u, w^u, \mathbf{x}^u)$ in the original menu, including $\cC^t$, with $\hat{\cC}^u = \bp{a^u, \pi^t + c(a^u), \mathbf{0}}$. For any type $u\le t$, the revenue from type $u$ in the original menu is equal to $\pi^u + c(a^u)$, and the IR constraint for type $u$ in the original menu says that type $u$'s value for the outcomes induced by $a^u$ is at least the revenue $\pi^u + c(a^u)$. Combined with $\pi^u \ge \pi^t$, this implies $$U(u; \hat{\cC}^u) \ge \pi^u + c(a^u) - \bp{\pi^t + c(a^u)} \ge 0,$$
    so the modified menu is IR. For types $u>t$, if $U(u;\hat{\cC}^u) < 0$ then we replace $\hat{\cC}^u$ with the opt-out option. For types $u\le t$, each contract $\hat{\cC}^u$ collects an upfront payment of $\pi^t + c(a^u)$ from the buyer and hence yields profit at least $\pi^t + c(a^u) - c(a^u) = \pi^t$ for the seller. Hence $$\Rupfront \ge \bp{\sum_{i=1}^t \mu^i} \cd \pi^t,\quad \fl t \iff \mu^t \cd \pi^t \le \Rupfront \cd \fr{\mu^t}{\sum_{i=1}^t \mu^i},\quad \fl t.$$ Summing over $t\in [T]$ yields $$\Pi \le \Rupfront \cd \sum_{t\in [T]} \fr{\mu^t}{\sum_{i=1}^t \mu^i} \le \Rupfront \cd H_\mu.$$
\end{enumerate}
\end{proof}

\cref{upfront-payment-only} shows that the worst-case gap in profit between implementing the two-part tariff versus using only upfront payments is characterized by $H_\mu$, and furthermore this multiplicative gap is tight. Since $H_\mu \ge H_T$, this multiplicative loss in profit is always at least $H_T = \Omega(\log T)$, which is unbounded as the number of buyer types $T$ increases. It is thus crucial for the seller to incorporate usage payments into their contracts.

\begin{remark}
    We have established that outcome-dependent usage payments are crucial for maximizing seller profit in the service provider problem. A natural question is whether they also increase seller profit in other models. Interestingly, the answer turns out to be negative for the lottery pricing problem, which can be viewed as an instance of the service provider problem but with infinitely many actions that induce every possible outcome distribution. We prove in \cref{a:usage-payment-lottery-pricing} that the maximum seller revenue in lotteries is the same with or without usage payments.
\end{remark}

\subsection{Upfront payments are necessary} \label{upfront-payment-necessary}

\newcommand{\Rusage}{\Pi_{\tt{usage}}}

To prove that upfront payments are needed to achieve maximum seller profit, we consider a modified version of the service provider problem where all upfront prices are set to 0. Let $\Rusage$ denote the maximum seller profit achievable through usage payments only. Then we have the following result, proven in \cref{a:usage-payment-only}:

\begin{proposition} \label{usage-payment-only}
    There exists a service provider problem instance for which $\fr{\Ropt}{\Rusage} \ge \fr32$.
\end{proposition}

Note that the multiplicative gap of $\fr32$ in \cref{usage-payment-only} is not as large as the $H_\mu$ gap in \cref{usage-payment-necessary}. This could indicate that charging a buyer upfront is less important than charging for the buyer's usage, though deriving tight bounds for $\fr{\Ropt}{\Rusage}$ remains an open question.

\para{Incomparability of $\Rupfront$ and $\Rusage$.}

\cref{upfront-payment-only} and \cref{usage-payment-only} further show that the quantities $\Rupfront$ and $\Rusage$ are generally incomparable. In the proof of \cref{upfront-payment-only}, we constructed a profit-maximizing menu with only usage payments, whereas any menu with only upfront payments is suboptimal, so $\Rusage > \Rupfront$ for that problem instance. In the proof of \cref{usage-payment-only}, we constructed a profit-maximizing menu with only upfront payments, whereas any menu with only usage payments is suboptimal, so $\Rupfront > \Rusage$ for that problem instance.

\subsection{Voluntary usage subsumes mandatory usage} \label{voluntary-usage}

In this section we show that seller profit is always weakly higher when the seller allows the buyer to choose whether to use the realized outcome or not than when the seller forces the buyer to use and pay for every outcome. Recall that the former assumption is \emph{voluntary usage} and the latter assumption is \emph{mandatory usage}. We show a stronger result: when the seller can commit to an action, mandatory usage is a \emph{special case} of voluntary usage with zero usage prices. This is done by establishing a surprising equivalence: the maximum seller profit under mandatory usage is exactly equal to the maximum seller profit under voluntary usage with zero usage prices, which is the quantity $\Rupfront$ from \cref{usage-payment-necessary}.

\newcommand{\Rmandatory}{\Pi_{\tt{mandatory}}}

Consider a modified version of the service provider problem that has the same seller-buyer interaction protocol described in \cref{model} but where in the last step the buyer is forced to accept the realized outcome $q$ and pay the usage price $x^t_q$. Denote by $\Rmandatory$ the maximum seller profit in this \emph{mandatory usage} model. Our main result is the following:

\begin{theorem} \label{voluntary-usage-subsumes}
 Let $\mu \in \Delta_{[T]}$ denote the buyer type distribution. Then
 \begin{equation} \label{R-Rmandatory}
     \Rupfront = \Rmandatory \le \Pi_{\te{opt}} \le H_\mu \cd \Rmandatory,
 \end{equation}
where $H_\mu$ was defined in \cref{eq:H-mu}. Moreover, the bounds above are tight: there exist problem instances for which $\Pi_{\te{opt}} = \Rmandatory$ and problem instances for which $\Pi_{\te{opt}} =  H_\mu \cd \Rmandatory$.
\end{theorem}

To prove the theorem statement, it suffices to show that $\Rmandatory = \Rupfront$, as the rest of the inequalities follow immediately from \cref{upfront-payment-only}. To establish  that $\Rmandatory = \Rupfront$, we use the following \emph{payment redistribution} argument:
\begin{lemma}\label{mandatory-ic-claim}
    Assuming mandatory usage, any IC menu can be modified into an IC menu with only upfront prices such that buyer utilities and seller profit remain the same.
\end{lemma}
\begin{proof}
    As long as there exists a contract $\cC^u= (a^u, w^u, \mathbf{x}^u)$ with a positive usage price $x^u_q > 0$, we modify $\cC^u$ to $\hat{\cC}^u$ by setting $x'^u_q = 0$ and $w'^u = w^u + p^{a^u}_q x^u_q$. Because of the mandatory usage assumption, the decrease in usage price increases all buyer utilities by the same amount that the increase in upfront prices decreases them. Since buyer utilities remain the same, the modified menu is IC. Furthermore, the revenue from each type is redistributed in equal amounts from the usage price into the upfront price and hence remains the same. 
\end{proof}

\begin{proof}[Proof of \cref{voluntary-usage-subsumes}]
        By \cref{mandatory-ic-claim}, a profit-maximizing menu under mandatory usage without loss of generality has only upfront payments and no usage payments. Under zero usage payments, however, the distinction between mandatory usage and voluntary usage is irrelevant since the buyer will always accept the outcome as their nonnegative valuation will always be at least the zero usage price. More formally, the claim implies the following equivalence:
        \begin{itemize}
            \item For every menu under mandatory usage there exists a menu under voluntary usage that has only upfront payments and achieves the same profit.
            \item Every menu under voluntary usage that has only upfront payments is a menu under mandatory usage that achieves the same profit.
        \end{itemize}
       We conclude that $\Rmandatory = \Rupfront$ as desired. As for the tightness of the inequalities in \cref{voluntary-usage-subsumes}, note that from \cref{upfront-payment-only}, we know there exist problem instances with $\Pi_{\te{opt}} = H_\mu \cd \Rupfront = H_\mu \cd \Rmandatory$. Furthermore, any problem instance where every type has the same valuation vector is equivalent to an instance with a single buyer type, for which  we have $H_\mu = 1$, implying again from \cref{upfront-payment-only} that $\Pi_{\te{opt}} = \Rupfront = \Rmandatory$. 
\end{proof}

\begin{remark}
    Past works on contract design assume forced payments for each outcome regardless of whether the price exceeds the buyer's valuation of the outcome, for example selling hidden actions \citep{bernasconi2024agent}, or charge only a lump sum payment for the action without requiring usage payments, for example selling lotteries \citep{chen2015complexity}. A natural  question is why we should consider a two-part payment scheme that first charges an  upfront price but then allows buyers the freedom to accept or reject the realized outcome.  From a technical perspective, \cref{voluntary-usage-subsumes} proves that the maximum seller profit is \emph{always} weakly higher under voluntary usage than under mandatory usage, and furthermore sometimes significantly so, as characterized by the multiplicative factor $H_\mu\in [H_T, T)$. But also in practice, voluntary usage gives the buyer the freedom to choose and hence could be more attractive from a marketing standpoint.
\end{remark}

Our result may appear counterintuitive since one might think the seller should benefit when they have the power to enforce something, in this case the buyer's acceptance of the outcome.  While the mandatory usage requirement is beneficial {\em after} the buyer has chosen to participate, it reduces the likelihood that a buyer participates in the mechanism. Under mandatory usage the seller has to incentivize the buyer more to participate in the mechanism. Voluntary usage, on the other hand, reduces the buyer's risk from the uncertain nature of outcomes and hence can make buyers value a contract \emph{more}. This makes them more willing to pay a higher upfront price, leading to increased seller profit.

\section{Characterizing profit-maximizing menus} \label{characterizations}

Having established the necessity of the two-stage payment structure, we now focus on identifying the profit-maximizing menu (\cref{profit-maximization-direct-menu}). Solving this problem directly is challenging, since the feasible space of menus in \cref{profit-maximization-direct-menu} is huge, requiring a selection of a usage price for each contract and outcome pair. To simplify the search space, in this section, we characterize the structure of the profit-maximizing menu.  These characterizations play an important role in deriving our main results in \cref{complexity} and \cref{single-parameter}.
We begin with the following result that shows it suffices to set the usage price for each outcome and contract to be one of two values, either $0$ or $\infty$.

\begin{theorem} \label{two-usage-prices-suffice}
    Any IC menu can be modified so that all usage prices satisfy $x^t_q\in \bc{0, \infty}$ while leaving the seller's profit and the buyers' utilities unchanged.
\end{theorem}

We give a proof sketch here: the formal proof is in \cref{a:two-usage-prices-suffice}. Define $S^t \coloneq \bc{q: v^t_q \ge x^t_q}$ to be the set of outcomes that type $t$ accepts. Our idea is to \emph{redistribute} all usage prices for outcomes in $S^t$ into the upfront price. For each $q\in S^t$, we set $x^t_q = 0$ and increase the upfront price $w^t$ by $p^{a^t}_q x^t_q$, calling the new contract $\cC'^t$. In this way, the usage price decreases by exactly the same amount that the upfront price increases and hence $U(t;\cC^{t'}) = U(t;\cC^t)$. For all other types $u$, we have $U(u;\cC'^t) \le U(u;\cC^t)$ since $u$ now pays an additional upfront price of $p^{a^t}_q x^t_q$ while their utility from the decrease in usage price $x^t_q$ increases by at most $p^{a^t}_q x^t_q$. Hence $u$ will still choose $\cC^u$ so the menu remains IC. After applying the price redistribution for all contracts, the end result is a menu whose usage prices $x^t_q$ satisfy $x^t_q = 0$ or $x^t_q > v^t_q$. In the latter case we can simply increase the usage prices to $\infty$.

\begin{remark} \label{two-usage-prices-suffice-remark}
    Note that the $\infty$ price in \cref{two-usage-prices-suffice} is only for notational convenience and the proof similarly works when $\infty$ is replaced by any price that excludes all types from accepting the outcome, for example the maximum valuation $\max_{t,q} v^t_q$ of any type for any outcome. \cref{two-usage-prices-suffice} has the interpretation that without loss of generality we can view the seller's revenue as coming entirely from the upfront price as long as the seller can exclude buyers from some outcomes. This is similar in spirit to many software subscriptions like ChatGPT for which various subscription tiers provide access to models with differing capabilities.
    We emphasize that even though no revenue is collected from usage prices in $\bc{0,\infty}$, having the usage prices is still necessary by \cref{upfront-payment-only}.
\end{remark}

To summarize, \cref{two-usage-prices-suffice} shows that there exists a profit-maximizing mechanism in which the seller collects upfront prices for various actions but limits the buyer's access to certain outcomes. For the rest of the paper we assume that all usage prices are in $\bc{0,\infty}$, so the revenue from type $t$ is simply the upfront price $w^t$. The service provider problem, \cref{profit-maximization-direct-menu}, simplifies to the following:

\begin{tcolorbox}[title=Maximizing profit of a direct menu with usage prices in $\bc{0,\infty}$]
\vspace{-1em}
    \begin{align}
             \max_{ \substack{\{(a^{t'}, w^{t'}, \mathbf{x}^{t'}) \}_{{t'} \in [T]} \\ x^{t'}_q\in \bc{0,\infty}}} & \quad \bE_{t \sim \mu} \left[ w^t  - c(a^t)
    \right] \label{binary-optimization-problem} \\ 
        & U(t; \cC^t) \geq U(t; \cC^{t'}) && \fl t, t' \in [T]   \notag \\ 
        & U(t; \cC^t) \geq 0  && \fl t  \in [T]  \notag
          \end{align}   
\end{tcolorbox}

Next, we show that in any profit-maximizing menu it is not necessary to set the usage price of any outcome to be $\infty$ in the contracts chosen by the types that pay the highest upfront price to the seller. To state our result, we make the following definition:

\begin{definition}[Highest-revenue types] \label{highest-revenue-definition} 
    For a menu $\mathcal{M}$ with $x^t_q \in \{0, \infty\}$ for each contract $t$ and each outcome $q$, a type $t$ is a \emph{highest-revenue} type if it pays the largest upfront price, i.e., $t \in \argmax_{s} w^{s}$.
\end{definition}

Note that, in general, more than one type can be a highest-revenue type. The following result implies that in the profit-maximizing menu, the highest-revenue types are never barred from using the realized outcome.

\begin{proposition} \label{highest-type-no-usage-prices}
    Any IC menu with usage payments in $\bc{0,\infty}$ can be modified so that the contracts for the highest-revenue types have all usage prices equal to zero while remaining IC and weakly increasing seller profit.
\end{proposition}

\begin{proof}
    For each highest-revenue type $t$, we replace type $t$'s contract $\cC^t = \bp{a^t, w^t, \mathbf{x}^t}$ with $\cC'^t = \bp{a^t, w^t, \mathbf{0}}$. Since usage prices decrease from $\cC^t$ to $\cC'^t$, we have $U\bp{t; \cC'^t} \ge U\bp{t; \cC^t}$, so type $t$ will choose $\cC'^t$ over any other contract $\cC^u$ and thus yield the same seller revenue $w^t$. Any other type $u\neq t$ will choose either $\cC^u$ or $\cC'^t$ in the modified menu, yielding seller revenue $w^u$ or $w^t \ge w^u$, respectively. In the latter case we replace $\cC^u$ by $\cC'^t$ in the modified menu to maintain IC, noting that seller profit weakly increases. 
\end{proof}

\begin{remark}
     \cref{highest-type-no-usage-prices}, which shows that the highest-revenue types do not require usage payments is similar in flavor to a result by \citet{bergemann2018design} that shows in the context of selling information that in any optimal menu of statistical experiments, the highest buyer types purchase the fully informative experiment. A similar phenomenon of allowing higher buyer types to receive more information is reflected in \citet{liu2021optimal}.
\end{remark}
\section{Complexity of computing a profit-maximizing menu} \label{complexity}

Using the characterizations of the profit-maximizing menu obtained in the previous section, we are now ready to analyze the computational complexity of identifying such a menu. We first establish in \cref{hardness} that computing an exact profit-maximizing menu is \tbf{NP}-hard even in a seemingly simple setting of two buyer types and a single seller action. In light of this negative result, in \cref{fptas-section}, we consider the problem of \emph{approximating} the profit-maximizing menu. For that problem, we obtain a positive result: we show there exists a fully-polynomial time approximation scheme (FPTAS) for maximizing seller profit when the number of buyer types $T$ is constant and when the profit is at least a constant fraction of the cost. To elaborate, for a fixed number of buyer types $T$ and any constant $\eps>0$, we can  find a menu that achieves at least a $1-\eps$ fraction of the maximum seller profit in time polynomial in $\fr{1}{\eps}$, the number of seller actions, and the number of outcomes.

\subsection{\tbf{NP}-hardness for two types and a single action} \label{hardness}

We show that computing the exact maximum seller profit is \tbf{NP}-hard even when there are just two buyer types and just a single seller action. (Maximizing seller profit for one buyer type is trivial because for each seller action we can extract maximum profit by setting an upfront price equal to exactly the buyer's value so that there is no buyer surplus.)

\begin{theorem} \label{np-hardness-two-types}
Computing the exact maximum seller profit in the service provider problem is \tbf{NP}-hard even when there are only two buyer types and a single seller action.
\end{theorem}

To establish this hardness result, in \cref{sec:two-type-formula} we first derive explicit expressions for the profit-maximizing upfront prices for two types given usage prices $x^t_q \in \{0, \infty\}$ for the case where the two types are equally likely. Using these expressions, we then show that the \tsf{Partition} problem can be reduced to an instance of the profit-maximization problem. In particular, we prove that a multiset of integers summing to $M$ can be partitioned into two subsets of equal sum if and only if the maximum seller profit in a specific instance is exactly $\fr{9M}{4}$.

\paragraph{Reducing from \tsf{Partition}.}

To prove \tbf{NP}-hardness of computing the optimal menu for two types, we reduce from the well-known \tsf{Partition} problem, as described below:
\begin{problem}[\tsf{Partition}]
    Given a multiset of integers $\bc{n_1,n_2,\lds, n_k}$ with sum $M = n_1 + n_2 + \cds + n_k$, determine if there exists a subset that sums to $\fr{M}{2}$.
\end{problem}

It is well-known that \tsf{Partition} is \tbf{NP}-hard, for example, see \citet{hayes2002easiest}. Given an instance $\bc{n_1,n_2,\lds,n_k}$ of the partition problem, we construct an instance of the service provider problem as follows:
\begin{itemize}
    \item Let $Q = 0\cup [k]$, $T = [2]$, $\mu^1 = \mu^2 = \fr12$ and $A = \bc{a}$.
    \item Outcome $0$ has valuations $v^1_0 = M(k+1)$ and $v^2_0 = 0$.
    \item For $q\in [k]$, outcome $q$ has valuations $v^1_q = n_q(k+1)$ and $v^2_q = 3n_q(k+1)$.
    \item The single action $a$ has cost $c(a) = 0$ and the transition probabilities to the $k+1$ outcomes are uniform, so $p^a_q = \fr{1}{k+1}$ for all $q$.
\end{itemize}

The maximum possible type 1 revenue is $\sum_q p^a_q v^1_q = 2M$ and the maximum possible type 2 revenue is $\sum_q p^a_q v^2_q = 3M.$ The key claim in the reduction, proven in \cref{a:partition-revenue} is the following:

\begin{claim} \label{partition-revenue}
    There exists a subset of $\bc{n_1,n_2,\lds,n_k}$ in the \tsf{Partition} problem that sums to $\fr{M}{2}$ if and only if the maximum seller profit in the corresponding service provider instance is $\fr{9M}{4}$.
\end{claim}

The claim shows how to reduce the \tsf{Partition} problem to computing the maximum seller profit of an instance of the service provider problem with two buyer types and a single action. Since \tsf{Partition} is \tbf{NP}-hard, computing the maximum seller profit in the service provider problem is also \tbf{NP}-hard. Certainly this also proves that computing a profit-maximizing menu that achieves this maximum profit is also \tbf{NP}-hard since computing the numerical profit that a menu achieves can be done in polynomial time.

\subsection{An FPTAS for constant number of types} \label{fptas-section}

As a complement to the hardness result in the preceding subsection, we now consider the problem of approximating the profit-maximizing menu. We establish the following positive result, providing an FPTAS for the setting where the number of buyer types is held constant and when the profit is at least a constant fraction of the cost:

\begin{theorem} \label{fptas}
    Assume that the number of buyer types $T$ is constant and that the maximum seller profit $\Pi_{\tt{opt}}$ is at least a constant fraction $\kappa>0$ of the costliest action performed by the seller in a profit-maximizing menu. Then for any $\eps > 0$ there is an algorithm that runs in time polynomial in $\fr{1}{\eps}$, $\ab{A}$, and $\ab{Q}$ and outputs a menu achieving profit $\Pi \ge (1-\eps)\cd \Pi_{\tt{opt}}$.
\end{theorem}

\begin{remark}
    In the above assumption, the profit needs only be a constant fraction of the costliest action that appears in a profit-maximizing menu as opposed to the costliest action in $A$. This assumption is also reasonable in practice: it is equivalent to the seller's profit margin (ratio of profit to revenue) being bounded away from zero. Finally, we argue why this assumption is necessary for an FPTAS to exist. Consider modifying the problem instance we constructed in \cref{hardness} so that that single action $a$ has cost $c(a) = \fr{9M}{4}$. By \cref{partition-revenue}, $\Pi_{\tt{opt}} = 0$ if and only if there exists a subset in the \tsf{Partition} problem summing to $\fr{M}{2}$, so determining whether $\Pi_{\tt{opt}} = 0$ is $\tbf{NP}$-hard. Any algorithm that outputs a menu achieving $(1-\eps)$-fraction of the maximum seller profit $\Pi_{\tt{opt}}$ necessarily has strictly positive profit when $\Pi_{\tt{opt}} > 0$ and strictly negative profit when $\Pi_{\tt{opt}} < 0$, hence approximating $\Pi_{\tt{opt}}$ is at least as hard as the \tbf{NP}-hard problem of determining whether $\Pi_{\tt{opt}} = 0$.
\end{remark}

\cref{fptas} implies that in situations where the number of buyer types is not too large, the seller can efficiently find a good approximation to the profit-maximizing menu. The rest of this section is devoted to proving \cref{fptas}. We start by fixing the action $a^t$ for each contract $\cC^t$ in the menu and find (approximately) profit-maximizing upfront and usage prices for this fixed mapping of contracts to actions. The overall (approximately) profit-maximizing menu can then be obtained by iterating over all $\ab{A}^T$ such possible mappings and outputting the most profitable among them. Since $T$ is assumed to be fixed, this entire procedure can be completed in polynomial time if each iteration can be. Under the assumption that $\Pi_{\tt{opt}}$ is at least a constant fraction $\kappa>0$ of costliest $c$ incurred by the seller in a profit-maximizing mapping, we claim that to $(1-\eps)$-approximate the maximum profit it suffices to $\bp{1 - \fr{\eps \cd \kappa}{1 + \kappa}}$-approximate the quantity $\Pi'_{\tt{opt}} = \Pi_{\tt{opt}} + c$. Indeed, the profit of a menu that $\bp{1 - \fr{\eps \cd \kappa}{1 + \kappa}}$-approximates $\Pi'_{\tt{opt}}$ is
\begin{align*}
    \Pi &\ge \bp{1 - \fr{\eps \cd \kappa}{1 + \kappa}} \cd \bp{\Pi_{\tt{opt}} + c} - c
    \\&\ge \bp{1 - \fr{\eps \cd \kappa}{1 + \kappa}} \cd \Pi_{\tt{opt}} - \fr{\eps \cd \kappa}{1 + \kappa} \cd c
    \\&\ge \bp{1 - \fr{\eps \cd \kappa}{1 + \kappa}} \cd \Pi_{\tt{opt}} - \fr{\eps \cd \kappa}{1 + \kappa} \cd \fr{\Pi_{\tt{opt}}}{\kappa}
    \\&\ge (1-\eps)\cd \Pi_{\tt{opt}}.
\end{align*}
Letting $\eps' := \fr{\eps \cd \kappa}{1 + \kappa}$, it suffices to construct an algorithm that outputs an menu $\cM$ with profit $\Pi$ satisfying $$\Pi' := \Pi + c \ge (1-\eps') \cd \bp{\Pi_{\tt{opt}} + c}$$ in time polynomial in $\fr{1}{\eps'}$, noting that $\eps' = \Omega(\eps)$. Our reason for adding $c$ to the profit will be apparent later as a condition to ensure nonnegativity of objective values for intermediate solutions of a certain recursive program. Since $\Pi_{\tt{opt}}$ is only assumed to be a constant fraction of $c$ rather than $\max_{a\in A} c(a)$, our algorithm, which recall iterates over all contract to action mappings, is only guaranteed to output a $(1-\eps)$-approximate solution when run on a profit-maximizing mapping. For notational brevity we denote $\Pi' = \Pi + c$ by $\Pi$ and $\eps' = \fr{\eps \cd \kappa}{1 + \kappa}$ by $\eps$ for the rest of the proof.

Recall that by \cref{two-usage-prices-suffice}, a profit-maximizing menu has usage prices $x^t_q\in \bc{0,\infty}$ for each $t \in [T]$ and outcome $q \in Q$. Since there are $2^{T\cd \ab{Q}}$ possible choices for the usage prices, searching over all of them is prohibitive. Instead, our proof exploits the structure of the problem to reduce the search space to polynomial size by skipping over similar choices of the usage prices. In particular, we use a general framework due to \citet{woeginger2000} to obtain an FPTAS for a class of problems with a recursive structure. (Woeginger refers to this class of problems as \emph{dynamic programs}: these are different from and more general than dynamic programming problems. To avoid confusion we refer to these as \emph{recursive programs} in this paper.) We present this general framework in \cref{sec:woeginger}. Next, in \cref{sec:recasting}, we show that the service provider problem can be recast into this framework by appropriately defining the underlying recursive structure. Finally, we show that the recast service provider problem satisfies the technical conditions required in \citet{woeginger2000} to admit an FPTAS.

\subsubsection{A general recursive program framework.} \label{sec:woeginger}

We begin by introducing a recursive program framework by \citet{woeginger2000} that considers optimization problems of the following form: There is a set of (potential) \emph{states} $S \subseteq \mathbb{R}^d$, and a function $\Pi \colon S \to \mathbb{R}$ defined over $S$. In addition, we have the following:

\begin{itemize}
    \item A set $S_0 \subseteq S$ of initial states.
    \item A sequence of inputs $z_1,\lds,z_n \in Z$.
    \item A finite set $\cF$ of \emph{transition functions}, where each function $f\in \cF$ maps a tuple $(s, z) \in S \times Z$ to a new state $f(s, z) \in S$.
\end{itemize}

For $i=1, \lds, n$, define $S_i = \{ f(s, z_i) : s \in S_{i-1}, f \in \cF\}$. The objective of the optimization problem is to compute $\max_{s \in S_n} \Pi(s)$, and to identify a sequence $f_1, \lds, f_n \in \cF$ such that there exists a tuple $(s_0, \lds, s_n)$ with $s_0 \in S_0$ and $s_i = f_i(s_{i-1}, z_i)$ for $i =1, \lds, n$ where $s_n$ achieves the maximum.
The recursive structure of the above problem follows from the fact that each state in $S_i$ encodes a partial solution to the problem using inputs $z_1,\lds,z_i$. Thus, an optimal solution can be built recursively from the partial solutions. In particular, we have
\begin{align*}
    \max_{s \in S_n} \Pi(s) = \max_{f \in \cF} \bc{\max_{s \in S_{n-1}} \Pi( f(s, z_n))}
\end{align*}

	

	

Note that, at this level of generality and without any additional assumptions on the objective function or the transition functions, solving the above problem recursively would require enumerating over all transition functions and the states. Since the number of states is multiplied at every step by the number $|\cF|$ of transition functions, such an approach will run in exponential time in general. In light of this, \citet{woeginger2000} identified three conditions under which the above problem admits an FPTAS. These conditions essentially require a form of continuity of the objective function and the transition functions. To introduce these conditions, we first define a notion of closeness of states:
\begin{definition}
    For a constant $r>1$, two states $s$ and $s'$ are \emph{$r$-close} if for each $i$, $r^{-1} \cd s_i \le s'_i \le r\cd s_i$.
\end{definition}
Using this notion of closeness, we consider the following assumption on the optimization problem:
\begin{assumption}\label{as:woeginger} Suppose the following conditions hold:
\begin{enumerate}
    \item \emph{Closeness preserved by transition functions:} For any transition function $f \in \cF$, input $z_i \in Z$, states $s, s'$, if $s$ and $s'$ are $r$-close then $f(s, z_i)$ and $f(s', z_i)$ are $r$-close.

    \item \emph{Continuity of the objective function:} There exists an integer $G\ge 0$ such that for any $1 < r \le 1 + \frac{1}{G+1}$, if states $s$ and $s'$ are $r$-close then $\Pi(s') \ge r^{-G} \cd \Pi(s)$.

    \item \emph{Computability.} All transition functions $f\in \cF$ and the objective function $\Pi$ can be evaluated in polynomial-time.
\end{enumerate}
\end{assumption}

We then have the following result:
\begin{theorem}[\citet{woeginger2000}] \label{thm:woeginger}
    Suppose \cref{as:woeginger} holds. Then there exists an algorithm $\cA$ that runs in time $\poly\mathopen{}\bp{\fr{1}{\eps}}$ and outputs a sequence $f_1,\lds,f_n$ of transition functions that results in a solution $s^* \in S_n$ such that $\Pi(s^*) \ge (1-\eps)\cd \max_{s\in S_n} \Pi(s)$.
\end{theorem}

The FPTAS in the preceding theorem, displayed as Algorithm \ref{dp-fptas}, is based on the idea of \emph{trimming} the state space. Instead of exploring exponentially many states corresponding to exponentially many choices of transition functions, a recursive program whose goal is to compute a $(1-\eps)$-approximation of the maximum objective value need only store a $\poly\bp{\fr{n}{\eps}}$-sized subset of states at every step by merging state vectors where the individual components of the states are within a $(1\pm \eps)$-multiplicative factor of each other. The objective function's proximity property guarantees that the recursive program outputs a $(1-\eps)$-approximation of the optimal objective value.


\begin{algorithm}[t]
    \caption{Converting recursive program to FPTAS}
	\label{dp-fptas}
    \SetAlgoNoLine
	\KwIn{A sequence $z_1,\lds,z_n$ of length $n$}
    
    \textbf{Parameters:} $\eps$ (the required approximation ratio), $G$ (the integer that governs proximity between states), $M$ (the maximum value that can appear in a component of a state)

    \textbf{Initialization:} Define $r := 1 + \fr{\eps}{2Gn}$ and $L := \ceil{\fr{M}{\log r}}$, noting that all values in a state vector are in the range $[0, r^L]$.

    \begin{enumerate}
        \item Partition the range $[0, r^L]$ into $L+1$ \emph{$r$-intervals} $I_0 = [0], I_1 = [1, r), \lds, I_L = [r^{L-1}, r^L]$.
        \item Partition the state space into \emph{$r$-boxes} where each coordinate is partitioned into some $I_\ell$, noting that if two states are in the same $r$-box they are $r$-close.
        \item Define $R := $ number of $r$-boxes, which is polynomial in the size of the input and in $\fr{1}{\eps}$.
    \end{enumerate}
        
	\textbf{Algorithm:} Let $T_0 \coloneq S_0 \coloneq $ the set of initial states. At each step we trim the state set $S_i$ into a smaller set $T_i$ that contains exactly one state in each $r$-box.
	
	\For{$i = 1,\lds,n$
		}{
			Let $S_i \coloneq \bc{f(s,z_i): f\in \cF, s\in T_{i-1}}$
            
            Let $T_i \coloneq $ a trimmed copy of $S_i$ (for each $r$-box that contains multiple states in $S_i$, keep exactly one state)
		}

        \KwOut{$s^* = \argmax \bc{\Pi(s):s\in T_n}$ and the sequence $f_1,\lds,f_n$ of transition functions achieving $s^*$.}
    \textbf{Runtime:} Polynomial in the number of possible states in each $T_i$, which is at most $R$.
\end{algorithm}

\subsubsection{Recasting the profit-maximization problem.} \label{sec:recasting} 

To cast the service provider's problem in the general framework presented above, we need to identify an appropriate notion of a state vector, an objective function defined at each state vector, and the transition functions. Furthermore, to obtain an FPTAS, we must ensure that the three conditions in \cref{as:woeginger} hold.

\paragraph{States:} Recall that the service provider's problem requires computing both the usage prices $\mathbf{x}^t$ and the upfront prices $w^t$ for each contract $t$ in the menu. We let the state correspond to the collection of usage prices $\mathbf{x}^t$ for each contract $t$; later on, we define the objective function as the maximizing the seller's profit over all upfront prices for a given collection of usage prices. However, to apply \cref{thm:woeginger}, we require a better \emph{state representation}, since the collection of usage prices lies in the set $\{0, \infty\}^{T\cd \ab{Q}}$ . We first prove that the collection of usage prices can be concisely represented by a state vector in $\mathbb{R}^{T^2}$. To do this, we first introduce the following definition:

\begin{definition}[Value of a contract]
    Fix the usage prices $\mathbf{x}^t \in \{0, \infty\}^{\ab{Q}}$ for each contract $t$. The \emph{value} of contract $t'$ for the buyer of type $t$ is defined as 
    \begin{align*}
    V(t;t') \coloneq \sum_{q} p^{a^{t'}}_q v^t_q \cd \one\bb{x^{t'}_q = 0}.
    \end{align*}
\end{definition}

For any choice of the upfront prices $(w^t)_{t \in [T]}$, we have the utility of buyer $t$ for the contract $t'$ equals $U(t;t') = - w^{t'} + V(t;t')$  Thus, the value of a contract can be interpreted as the buyer's utility \emph{ignoring} the upfront price.

Given $U(t;t') = - w^{t'} + V(t;t')$, the IC and the IR constraints can be rewritten respectively as

\begin{align*}
    w^t - w^{t'} &\leq  V(t; t) - V(t; t') \quad \text{for all $t, t'$.}\\
    w^t &\leq V(t; t) \quad \text{for all $t$.}
\end{align*}

Thus, the set of feasible upfront prices depends on the usage prices $(\mathbf{x}^t)_{t \in [T]}$ only via the contract values $(V(t;t')_{t, t' \in [T]}$. Furthermore, the seller's profit from any IC and IR menu depends only on the upfront prices. Together, this implies that the contract values $(V(t;t')_{t\in [T]}$ can be used to concisely represent the usage prices. 

Hereafter, we let $S = \mathbb{R}^{T^2}$ denote the set of (potential) states. Each choice of usage prices $(\mathbf{x}^t)_{t \in [T]}$ corresponds to a state $s \in S$ with 
\begin{align*}
    s_{t, t'} = V(t;t') = \sum_{q} p^{a^{t'}}_q v^t_q \cd \one\bb{x^{t'}_q = 0} \quad \text{for each $t,t'$}.
\end{align*}

\paragraph{Objective function:} Given the definition of the state, a potential choice for the objective function is the optimal seller's profit across all upfront prices for an IC and IR menu with the corresponding contract values. Using the IC and the IR constraints listed above, this objective can be written as the following linear program:

\begin{tcolorbox}[title=Maximizing profit in direct menus fixing actions and values]
\vspace{-1em}
    \begin{align} 
        \Pi_{\tt{direct}}(s) = \max_{(w^t)_{t\in [T]}} &  \bE_{t\sim \mu} \bb{w^t - c(a^t)} + c \label{linear-program-given-usage-prices} \\
        & w^t \le w^{t'} + s_{t,t} - s_{t,t'} & \fl t,t' && \te{(IC)} \notag \\
            & w^t \le s_{t,t}  & \fl t && \te{(IR)} \notag
\end{align}
\end{tcolorbox}

However, the objective function $\Pi_{\tt{direct}}(s)$ in \cref{linear-program-given-usage-prices} is not continuous in $s$. This is because a small change in a value $s_{t, t'}$ could result in the menu of contracts no longer being feasible because of the violation of a IC constraint, so $\Pi_{\tt{direct}}(s)$ could suddenly jump from a finite value to $-\infty$ due to infeasibility. Using $\Pi_{\tt{direct}}$ as the recursive program objective would violate the continuity requirement in \cref{as:woeginger} and hence would fail to yield an FPTAS.

To address this issue, we instead consider the objective at any state to be the maximum seller revenue across all upfront prices for \emph{any} menu with corresponding contract values. In particular, we expand the space of menus to include all indirect menus, by relaxing the IC and IR requirements. While unintuitive,  this modification ensures that the objective function satisfies the continuity requirement in \cref{as:woeginger}, as we prove later in \cref{value-proximity}.

To formally define the objective function, we first develop some notation to capture the opt-out option. Define $s_{t,0} = 0$ and $w^0 = 0$ for all $t \in [T]$, and introduce an action $a^0$ with cost $c(a^0) =0$. In other words, opting-out is equivalent to choosing a contract that has zero value (for all buyer types) and zero upfront payment, and which results in the seller choosing a trivial zero-cost action.  With this notation in place, we define the objective function $\Pi_{\tt{indirect}}(s)$ over any state $s \in S$ as follows:

\begin{tcolorbox}[title=Maximizing profit in indirect menus fixing actions and values]
\vspace{-1em}
    \begin{align*}
             \Pi_{\tt{indirect}}(s) &= \max_{(w^t)_{t\in [T]}}\bE_{t \sim \mu} \left[ w^{u_\mathbf{w}(s,t)} - c(a^{u_\mathbf{w}(s,t)}) \right] + c \\
            &\quad \text{where }  u_\mathbf{w}(s, t) \coloneq  \argmax_{t'\in \bc{0} \cup [T]} s_{t, t'} - w^{t'}
          \end{align*}   
\end{tcolorbox}

In the above formulation, rather than requiring that a buyer of type $t$ choose the contract $\cC^t$ as we do in a direct menu, we allow buyers to choose their utility-maximizing contract from the menu, or opt-out if no contract offers nonnegative utility for the buyer. Here, for any choice of upfront prices $\mathbf{w} = \bp{w^t}_{t\in [T]}$, the expression $u_\mathbf{w}(s,t)\in \bc{0}\cup [T]$ specifies the optimal choice of a buyer of type $t$, when the contract values are given by $s$. In particular, if $u_\mathbf{w}(s,t) = t' \in [T]$, then the contract $\cC^{t'}$ maximizes the utility of a buyer of type $t$, whereas $u_\mathbf{w}(s,t) = 0$ implies that the buyer of type $t$ is better opting-out. As often common in contract-design problems, we assume that the buyer breaks ties in favor of the seller.

Note that in general, the functions $\Pi_{\tt{direct}}(s)$ and $\Pi_{\tt{indirect}}(s)$ take different values for a given state $s$. By restricting attention to upfront prices $\mathbf{w}$ that satisfy the IC and IR conditions, it follows that $\Pi_{\tt{indirect}}(s) \geq \Pi_{\tt{direct}}(s)$. Nevertheless, from a revelation principle argument, it follows that for each $s \in S$, there exists an $s' \in S$ such that $\Pi_{\tt{direct}}(s') = \Pi_{\tt{indirect}}(s)$. Thus, the maximum of the two functions are equals.



Indirect menus do not have the discontinuous IC conditions present in direct menus. That being said, it is not clear that $\Pi_{\tt{indirect}}$ is continuous in the state $s$ because the buyer's optimal response $u_{\mathbf{w}}(s,t)$ can jump between contract. A small change in a value $s_{t,t'}$ could lead to a buyer choosing an entirely different contract with a significantly higher or lower upfront payment. Despite this potential discontinuity in buyer's response, we show that the function $\Pi_{\tt{indirect}}(s)$ satisfies the continuity requirements in \cref{as:woeginger}: 

\begin{lemma}[Continuity of the objective function] \label{value-proximity}
    There exists an integer $G\ge 0$ such that if states $s$ and $s'$ are $r$-close for any $r < 1 + \fr{1}{G+1}$ then $\Pi_{\tt{indirect}}(s') \ge r^{-G}\cd \Pi_{\tt{indirect}}(s)$. 
\end{lemma}

Our proof of \cref{value-proximity} in \cref{a:value-proximity} relies on the following technical result: for two states $s$ and $s'$ that differ in only one component $i = (t, t')$ by a constant $\eps>0$, we have $\Pi_{\tt{indirect}}(s') \ge \Pi_{\tt{indirect}}(s) - \eps$. We also show through an enumeration argument that the computability requirement in \cref{as:woeginger} is also satisfied by $\Pi_{\tt{indirect}}(s)$. The proof is provided in \cref{a:value-efficient}.

\begin{lemma}[Computability] \label{value-efficient}
    When the number of buyer types $T$ is constant, the seller profit $\Pi_{\tt{indirect}}(s)$ is efficiently computable given $s$.
\end{lemma}

\paragraph{Inputs and transition functions.} Having defined the states and the objective function, we now describe the appropriate notion of \emph{inputs} and \emph{transition functions} to cast our problem into the framework of \citet{woeginger2000}. 

We let the inputs $z_1,\lds,z_n \in Z$ be an arbitrary ordering of the $n = T\cd \ab{Q}$ type and outcome pairs $(t,q)$. The algorithm proceeds as follows: At step $0$, all usage prices $x^t_q$ are set at $\infty$, corresponding to the initial state $s = \zero$. At each step $i \geq 1$, the input $z_i = (t_i, q_i)$ indexes the type-outcome pair whose usage price $x^{t_i}_{q_i}$ will be decided; the decision is either to keep its value at $\infty$ or to reduce it to $0$. Each decision influences the current value of the state, which can be represented in the form of two transition functions $f \in \{f_\infty, f_0\}$. The transition function $f_\infty$ corresponds to the decision of not reducing the usage price from $x^t_q = \infty$; in this case, the state does not change, and hence we have $f_\infty(s, (t,q)) = s$. The transition function $f_0$ corresponds to reducing the usage price to $x^t_q = 0$; in this case, the state updates to $s' = f_0(s,(t,q))$ given by
$$s'_{t',u} = \begin{cases} s_{t',u} + p^{a^u}_q v^{t'}_q & \text{if $u = t$}\\ s_{t',t} & \text{otherwise}, \end{cases} \quad \text{for all $t', u \in [T]$.}$$
The above transition captures the fact that if we reduce the usage price to $x_q^t = 0$, the value of the contract $t$ for each buyer of type $t'$ increases by $p^{a^t}_q v^{t'}_q$ (since under the contract $t$, the buyer will now use the service if the outcome $q$ is realized). When all the inputs are processed, the final state corresponds to the complete assignment of the usage prices for each type-outcome pair.

From the definition, it immediately follows that the closeness of states is preserved by each transition function, as required in \cref{as:woeginger}. 

\begin{lemma}[Closeness of states is preserved by the transition functions]\label{lem:closeness-states}
     For any input $z_i$, and states $s, s'$, if $s$ and $s'$ are $r$-close for some $r \geq 1$, then $f(s, z_i)$ and $f(s', z_i)$ are $r$-close for each $f \in \{f_\infty, f_0\}$.
\end{lemma}

\begin{proof}
    Both transition functions $f_\infty$ and $f_0$ increment the state by a vector with nonnegative entries that does not depend on the current state, and $r^{-1} \cd s_i \le s'_i \le r\cd s_i$ implies $r^{-1} \cd (s_i + \Delta) \le s'_i + \Delta \le r\cd (s_i + \Delta)$ for any $\Delta \ge 0$ and $r \geq 1$. 
\end{proof}

Taken together, \cref{value-proximity}, \cref{value-efficient}, \cref{lem:closeness-states} imply that all the conditions in \cref{as:woeginger} hold, and thus, \cref{dp-fptas} provides an FPTAS to our problem. This completes the proof of \cref{fptas}.

\section{The single-parameter setting} \label{single-parameter}

This section examines a fundamental special case of the general service provider problem in which each buyer type's valuation vector can be parametrized by single real-valued number. Such single-parameter preferences are well-studied in contract design problems \citep{alon2021contracts,li2022selling,chen2015complexity} as well as in broader economic settings \citep{mussa1978monopoly}.

\begin{definition}[Single-parameter setting] \label{single-parameter-definition}
 Buyer types are \emph{single-parameter} if each buyer type $t$'s valuation vector is characterized by a single parameter $\alpha^t>0$  such that $v^t_q = \alpha^t \cd v(q)$ for every $q$, where $v(q) \in \bR_{\ge 0}$ is the \emph{baseline} valuation for outcome $q$. 
 \end{definition}

Since buyers with the same single parameter $\alpha$ have the same valuations and hence can be merged into one type, we can without loss of generality label the buyer types so that $\alpha^1 < \alpha^2 < \cds < \alpha^T$. The main result of this section is that in the single-parameter setting, the maximum seller \emph{revenue} can be achieved by a single contract for all buyer types:

\begin{theorem} \label{single-parameter-revenue}
    The single-parameter setting admits a \emph{revenue}-maximizing menu that consists of a single contract for all buyer types. Moreover, this single contract can be computed in polynomial time in $T$,  $\ab{A}$, and $\ab{Q}$.
\end{theorem}

Observe that our proof of the \tbf{NP}-hardness of computing profit-maximizing menus (\cref{np-hardness-two-types}) involves a reduction of \tsf{Partition} to a service provider problem where the costs are zero, implying that in general, computing a revenue-maximizing menu is also \tbf{NP}-hard. Thus, the preceding result implies that the single-parameter assumption brings a computational advantage. Observe that if all actions have the same costs, the above result trivially implies that profit-maximizing menu can also be computed efficiently. However, this result fails to hold if different actions incur different costs. To show this, in \cref{a:no-single-contract-example} we construct a problem instance where actions have heterogeneous costs in which the profit of any menu with a single contract is strictly lower than the optimal profit.

In the remainder of this section, we prove \cref{single-parameter-revenue}. Our proof involves the following steps:
First, we show that any IC direct menu of contracts satisfies a monotonicity condition on the values $V(t;t)$ of contract $\cC^t$ for type $t$. We then use this characterization to derive exact formulas for the revenue-maximizing upfront prices in terms of the contract values. Third, we use the fact that these formulas are \emph{linear} in the contract values to relax the revenue-maximization problem in the single-parameter setting into a linear program in the buyer values. Finally, we show that any extreme point of the feasible region of this relaxed linear program is achievable by a menu consisting of a single contract. Since a linear program is maximized at an extreme point, we conclude that a single contract is revenue-maximizing. All missing proofs in this section are provided in \cref{a:higher-types-higher-values}.
 
\paragraph{Notation.} Define $S^t \coloneq \{q : x^t_q = 0\}$ and recall that $V(t;u) = \sum_{q\in S^u} p^{a^u}_q v^t_q$ is the value of contract $\cC^u$ for type $t$. For notational brevity, we define $V^t \coloneq V(t;t)$. The single-parameter assumption implies that $$V(t;u) = \sum_{q\in S^u} p^{a^u}_q v^t_q = \frac{\alpha^t}{\alpha^u} \cd \left(\sum_{q\in S^u} p^{a^u}_q v^u_q\right) = \frac{\alpha^t}{\alpha^u} \cd  V^u$$ and hence
\begin{align*}
    U(t;u) = V(t;u) - w^u = \fr{\alpha^t}{\alpha^u} \cd V^u - w^u.
\end{align*}
From an argument similar to the proof of the Myerson's Lemma, we obtain the following monotonicity result:

\begin{lemma} \label{higher-types-higher-values} 
    In any IC menu of contracts, we have $\fr{V^t}{\alpha^t} \geq \fr{V^u}{\alpha^u}$ and $w^t \geq w^u$ for all buyer types $t,u$ satisfying $t \ge u$.
\end{lemma}

Using this result, we derive expressions for the revenue-maximizing upfront prices for any given contract values, which in turn allows us to express the express the seller's revenue purely in terms of the contract values:

\begin{lemma} \label{revenue-maximizing-upfront-prices}
For given contract values $\bp{V^t}_{t\in [T]}$, the revenue-maximizing upfront prices are given by

\begin{align*}
    w^t =  V^t  - \sum_{i=1}^{t-1} \left( \alpha^{i+1} - \alpha^i\right) \frac{V^i}{\alpha^i} \quad \text{for all $t$}.
\end{align*}

The corresponding seller revenue is given by

\begin{align}
    \sum_{t\in [T]} \mu^t \cd w^t = \sum_{t\in [T]} \left(\alpha^t \mu^t -  \left(\alpha^{t+1} - \alpha^t\right) \cd \sum_{i=t+1}^T  \mu^i  \right) \frac{V^t}{\alpha^t}.
\end{align}

\end{lemma}

Having expressed the seller's revenue solely in terms of the contract values, we can now optimize over all valid contract values to maximize the revenue. However, doing this directly involves identifying the usage prices for each contract (which can take values in $\{0, \infty\}$), which can be challenging. Instead, we pose a linear programming relaxation to this problem, and then show that this relaxation is tight. Toward that goal, we first define 
$$M \coloneq \max_{a\in A} \sum_q p^a_q \cd v(q),$$ and note that 

\begin{equation} \label{eq:Vt-bound}
V^t =\sum_{q\in S^t} p^{a^t}_q v^t_q  \le  \alpha^t \cd \bp{\max_{a\in A} \sum_q p^a_q \cd v(q)} = \alpha^t \cd M.
\end{equation}

Using this, we define the following linear program:
\begin{tcolorbox}[title=Linear program relaxation for maximum revenue]
    \begin{align}
    \Pi_{\tt{relax}} := \max_y & \sum_{t\in [T]} \left(\alpha^t \mu^t -  \left(\alpha^{t+1} - \alpha^t\right) \cd \sum_{i=t+1}^T  \mu^i  \right) y^t   \label{relaxation}\\
    &  y^t \geq y^{t-1}, &&  \fl t  \notag\\
    & 0\le y^t \le M, && \fl t \notag
\end{align}
\end{tcolorbox}
Let $(V^t)_{t \in [T]}$ denote the contract values for an IC direct menu $\cM$, and define $y^t = \fr{V^t}{\alpha^t}$ for each $t$. Since $V^t \leq \alpha^t \cd M$, we see that $0 \leq y^t \leq M$, and from \cref{higher-types-higher-values}, we obtain that $y^t \geq y^{t-1}$. Thus, we obtain that all valid contract values are feasible for the linear program. From \cref{revenue-maximizing-upfront-prices}, the objective of the linear program is the same as the seller's revenue under the revenue-maximizing upfront prices. Thus, we conclude that the linear program is indeed a relaxation, and we have $\Pi_{\tt{relax}} \geq \Pi_{\tt{opt}}$. 

Note that not all feasible values $(y^t)_{t \in [T]}$ are achievable by the contract values of an IC menu. Nevertheless, we show that the optimal value of the linear program is indeed achieved by valid contract values, and hence $\Pi_{\tt{relax}} = \Pi_{\tt{opt}}$. This is because, the optimal value of the linear program \cref{relaxation} is achieved at some extreme point $(y^t)_{t \in [T]}$ of feasible region, which must of the form $y^t = 0$ for all $t < t^*$  and $y^t = M$ for all $t \geq t^*$ for some $t^* \in [T]$. This feasible point corresponds to values $V^t = \alpha^t y^t$, which can be achieved by offering a single contract $$\cC^{t^*} = \bp{\argmax_{a\in A} \sum_q p^a_q \cd v(q), M \alpha^{t^*}, \mathbf{0}}$$ to all buyer types. Indeed, buyer types $t$ with $t<t^*$ will opt-out of the mechanism and buyer types with $t \ge t^*$ will choose $\cC^{t^*}$.

To see that the revenue-maximizing single contract can be computed in polynomial time, we can iterate over all $t^*\in [T]$, compute the revenue of the single-contract menu $\cM^{t^*} = \bc{\cC^{t^*}}$, and output the menu $\cM^{t^*}$ that yields the highest revenue. This completes the proof of \cref{single-parameter-revenue}.
\section{Conclusion and open problems}

In this paper we describe a contract-based mechanism for selling a service that produces an outcome of uncertain quality. We believe that our modeling assumptions such as voluntary usage are closely aligned with how machine learning training providers sell their services. We show that the two-part payment scheme is rich enough to significantly increase seller profit compared to other natural payment structures, but not so complex as to make it intractable to solve for an approximately profit-maximizing menu for a constant number of buyer types. There are a few interesting future directions left open by our research.

\paragraph{Optimizing menus of constant size.} It would be interesting to study the computational complexity of computing a revenue-maximizing or profit-maximizing menu that consists of a constant $k$ number of contracts in the general service provider problem. This problem falls into a classic area of contract design that explores the trade-off between \emph{optimal} contracts and \emph{simple} contracts \citep{dutting2019simple, guruganesh2023menus, dutting2024algorithmic}. In many principal-agent problems, computing the optimal contract or menu of contracts in general is known to be \tbf{NP}-hard but computing a \emph{simple} menu, for example a linear contract or a menu with a constant number of contracts, can be done in polynomial time \citep{guruganesh2021contracts, castiglioni2022bayesian}.

In the context of our service provider problem, a natural definition of a \emph{simple} menu is one that consists of a small number of distinct contracts. Empirical studies have shown that customers may be negatively biased when presented with large sets of options, leading to reduced purchase behavior \citep{thaler2015misbehaving}. \citet{bernasconi2024agent} show that in their hidden action model, computing a profit-maximizing menu consisting of a constant $k$ number of contracts can be done in polynomial time. In our model the complexity of computing a profit-maximizing menu consisting of even a single contract is an open question, although we note that our recursive program framework in \cref{fptas-section} provides an FPTAS to \emph{approximate} the maximum seller profit achievable using a menu with $k$ contracts when both $k$ and the number of buyer types $T$ is constant. Motivated by our result that single-contract menus are revenue-optimal in the single-parameter setting (\cref{single-parameter-revenue}), we can ask in what other settings is a single contract approximately revenue- or profit-optimal?



\paragraph{Optimizing menus when the seller cannot commit.}
This paper assumes that the service provider can commit to actions, in part because automated machine learning platforms are regulated and in part because committing always leads to weakly higher profit for the provider. In contrast, \citet{bernasconi2024agent} adopt a model in which the seller action is \emph{hidden} and the buyer cannot assume that the seller will always perform the action specified by the contract. 

Mathematically, the inability to commit (and hidden actions) adds an additional IC constraint in the service provider problem: seller profit from their chosen action must be least as high as if they had switched actions. Since this constraint is linear in the upfront and usage prices, under \emph{mandatory usage} we can prove that the optimal menu can be computed efficiently by a linear program even with this seller IC constraint (this follows from \cref{voluntary-usage-subsumes}). However, the presence of this constraint creates challenges under \emph{voluntary usage}. Our proof of \textbf{NP}-hardness (\cref{np-hardness-two-types}) also shows that computing an optimal contract under hidden actions is \textbf{NP}-hard. This is because we construct a problem instance with a single action, implying that there is no other action to switch to. Unfortunately, our FPTAS result does not extend to this setting, as our $\{0,\infty\}$-usage price characterization result (\cref{two-usage-prices-suffice}) no longer holds. To see why, note that under committed actions the seller can redistribute payments between upfront and usage without changing seller or buyer behavior by \cref{two-usage-prices-suffice}. However, without commitment, if the usage payments are transferred upfront, the seller will simply perform the lowest-cost action. It would be interesting to see if some of our results, either on the necessity of a two-part payment scheme or on the computational complexity of computing an optimal menu, extend to settings where the seller cannot commit.

\newpage
\bibliographystyle{plainnat}
\bibliography{references}

\newpage
\appendix
\section{Supplemental proofs} \label{a}

\subsection{Necessity of two-part tariffs and voluntary usage} \label{a:H_mu} \label{a:usage-payment-only} \label{a:usage-payment-lottery-pricing}


\begin{lemma} \label{H_mu}
    The factor $H_\mu$ always lies in the range $[H_T, T)$, where $H_T = \sum_{t\in [T]} \fr{1}{t}$ is the $T$-th harmonic number, and furthermore this range for $H_\mu$ is tight.
\end{lemma}

\begin{proof}

    Without loss of generality, we assume $\mu^1 \leq \mu^2 \leq \cds\leq \mu^T$. If $\mu^1 > 0$,  we obtain
    \begin{align*}
        H_\mu = \sum_{t\in [T]} \frac{\mu^t}{\sum_{i=1}^t \mu^i}.
    \end{align*}
    If any $\mu^t = 0$, then we interpret $\frac{\mu^\tau}{\sum_{i=1}^\tau \mu^i}$ as $1$ for each $\tau \leq t$, and formally extend the definition of $H_\mu$ to all distributions $\mu$ on $[T]$ satisfying $0 \leq \mu^1 \leq \mu^2 \leq \cds\leq \mu^T$. Our goal is then to show that $H_\mu \in [H_T, T)$ for all such distributions.
    
    The upper bound $T$ follows trivially since each term is at most $1$. The upper bound can be achieved in the limit by setting $\mu^t =  \frac{\eps^{T-t} (1- \eps)}{1 - \eps^{T}}$ with $\eps \to 0$. Note that each term $$\frac{\mu^t}{\sum_{i = 1}^t \mu^i} = \frac{\eps^{T-t}}{\sum_{i=1}^t \eps^{T-i}} =  \frac{1 - \eps}{1 - \eps^{t}} \geq 1-\eps,$$ and hence $H_\mu \geq T(1-\eps) \to T$  as $\eps\to 0$. (Alternatively, the upper bound is attained by the distribution $\mu$ with $\mu^T=1$, under our extended definition.)

    We next show $H_\mu  \geq H_T$ using an induction argument. The inequality holds trivially when $T=1$. Suppose $T \geq 2$, and that the inequality holds for any distribution $\mu'$ on $[T-1]$. If $\mu^T = 1$, then $H_\mu = T \geq H_T$. On the other hand, if $\mu^T < 1$, then define $\bar{\mu}^t =  \fr{\mu^t}{1-\mu^T}$ for $t < T$. Note that $\bar{\mu}$ is a distribution on $[T-1]$. We have
    \begin{align*}
        H_\mu &= \sum_{t=1}^{T-1} \frac{\mu^t}{\sum_{i=1}^t \mu^i} + \mu^T\\
        &= \sum_{t=1}^{T-1} \frac{\bar{\mu}^t}{\sum_{i=1}^t \bar{\mu}^i} + \mu^T\\
        &\geq H_{T-1} + \mu^T\\
        &\geq H_{T-1} + \frac{1}{T}\\
        &= H_T.
    \end{align*}
    Here, the first equality follows from the definition and the second equality follows after dividing the numerator and the denominator of each term in the summation by $1-\mu^T > 0$. The first inequality holds by induction hypothesis, and the second inequality follows from the fact that $\mu^T = \max_t \mu^t \geq  \frac{1}{T}$. The last equality follows from the definition of the harmonic number. This completes the induction argument. Finally, it is easy to see that the lower bound is achieved  by setting $\mu^1 = \cdots = \mu^T = \frac{1}{T}$. 
\end{proof}

\subsubsection{Multiplicative loss in profit using only usage payments} 

\begin{proof}[Proof of \cref{usage-payment-only}]

    To show that $\fr{\Pi}{\Rusage} \ge \fr32$, we construct the following problem instance:
    \begin{itemize}
        \item Let $T = 2$, $\mu^1 = \mu^2 = \fr12$, and $Q = \bc{1, 2}$. There is a single action $a$ with cost $c(a) = 0$ and transition probabilities $p^a_q = \fr{1}{2},\fl q\in [2]$.
        \item Valuations are given by $\mathbf{v}^1 = \bp{1,\fr12}$ and $\mathbf{v}^2 = \bp{\fr12, 1}$.
    \end{itemize}

    Note that the utility of action $a$ for both types is $\fr34$, so $\Pi \le \fr34$. The upper bound on $\Pi$ can be achieved using the contract $\cC = \bp{a, \fr34, \mathbf{0}}$ for all types. On the other hand, we show that $\Rusage\le \fr12$. Assume for contradiction that $\Rusage > \fr12$. Then the revenue from some type is greater than $\fr12$; without loss of generality, assume this is type $1$. Since the revenue from type $1$ equals $$\frac{1}{2} \left( x_1^1 \cd \one\bb{ x_1^1 \leq 1} + x_2^1 \cd \one\bb{ x_2^1 \leq \frac{1}{2}} \right),$$ this implies that $0 \leq x_2^1 \leq \frac{1}{2}$, $x_1^1 + x_2^1 > 1$ and $\frac{1}{2} < x_1^1 \leq 1$.
    
    Note that, by incentive compatibility, $$\fr12 \bp{\fr32 - x^1_1 - x^1_2} = U(1;\cC^1) \ge U(1;\cC^2) \ge \fr12 \bp{1 - x^2_1},$$ which upon rearranging yields $x_1^2 \geq x_1^1 + x_2^1 - \frac12$. From $x_1^1 + x_2^1 > 1$, we then obtain $x_1^2 > \frac12$, implying that type $2$ does not use outcome $1$. If type $2$ does not use outcome $2$ as well, then we obtain $$\Rusage = \frac14 (x_1^1 + x_1^2) \leq \frac38 \leq \frac12,$$ a contradiction. On the other hand, if type $2$ uses outcome $2$, then we have 
    $$\fr12 \bp{1 - x^2_2} = U(2;\cC^2) \ge U(2;\cC^1) = \fr12 \bp{1 - x^1_2},$$ which implies $x^2_2 \le x^1_2 \le \fr12$. This yields $$\Rusage = \fr14\bp{x^1_1 + x^1_2 +  x_2^2} \le \frac14\bp{   \frac32+  \fr12} = \fr12,$$ a contradiction. We conclude that $\Rusage \le \fr12$. 
\end{proof}

\subsubsection{Usage payments are redundant in lottery pricing} 

    \begin{proposition} \label{usage-payment-lottery-pricing}
    In the lottery pricing problem, any menu with item-dependent payments can be modified into one with only upfront payments such that the buyer utilities and seller revenue remain unchanged.
\end{proposition}

At a very high-level, \cref{usage-payment-lottery-pricing} is true because the availability of all outcome distributions implies that to price discriminate between buyer types, the seller can choose actions that lead to different desired outcome distributions for different types. The power of being able to induce any distribution turns out to subsume the power of setting outcome-dependent usage prices. We now give the formal proof.

\begin{proof}[Proof of \cref{usage-payment-lottery-pricing}]
Because there is an action for each outcome distribution $\mathbf{p}$ in the lottery problem, specifying the action in a contract is equivalent to specifying the outcome distribution. We show that usage prices can be \emph{redistributed} into the lottery price. Let the menu consist of lotteries $(w^t, \mathbf{p}^t, \mathbf{x}^t)$ where $w^t$ is the lottery price, $p^t_q$ is the probability of receiving item $q$, and $x^t_q$ is the price of item $q$.

\paragraph{Mandatory usage.} We first consider mandatory usage in which the buyer, after being presented with item $q$ with probability $p^t_q$, is required to accept it at price $x^t_q$. Buyer type $t$ solves the optimization problem $$\max_{t\in [T]} -w^t+ \sum_{q\in Q} p^t_q \bp{v^t_q - x^t_q}$$ and the seller revenue from type $t$ is $$w^t + \sum_{q\in Q} p^t_q x^t_q.$$ We show how to construct an equivalent menu that does not have usage payments. By setting usage prices as $\bar{x}^t_q = 0,\fl q$ and the lottery payment $\bar{w}^t = w^t+ \sum_{q\in Q} p^t_q x^t_q$, we observe that the buyer's utility for each contract remains the same, hence type $t$ still chooses the same lottery, implying the seller revenue from type $t$ is unchanged as well.

\paragraph{Voluntary usage.} In voluntary usage, the buyer, after being presented with item $q$ with probability $p^t_q$, can decide whether or not to accept the item at price $x^t_q$. Buyer type $t$ solves the optimization problem $$\max_{t\in [T]} -w^t + \sum_{q\in Q} p^t_q \cdot \max\bc{0, v^t_q - x^t_q},$$ and the seller revenue from type $t$ is $$w^t + \sum_{q\in Q} p^t_q x^t_q \cdot \one\bb{v^t_q \ge x^t_q}.$$ We show how to construct an equivalent menu that does not have usage payments. Recall that by the opt-out assumption, there is a trivial item, which we label 0, such that $v^t_0,\fl t\in [T]$. For each $t\in [T]$ replace the lottery $(\mathbf{p}^t, w^t, \mathbf{x}^t)$ with $(\mathbf{p'}^t, {w'}^t, \mathbf{0})$ where
\begin{align*}
        {w'}^t &= w^t + \sum_{q: v^t_q \ge x^t_q} p^t_q x^t_q \\
        {p'}^t_q &= \begin{cases} p^t_q & \teif v^t_q \ge x^t_q \text{ and $q \neq 0$}  \\ 0 & \teif 0 < v^t_q < x^t_q \text{ and $q \neq 0$}  \\ 1 - \sum_{q: v^t_q \ge x^t_q} p^t_q & \teif q = 0. \end{cases}
\end{align*}
In other words, for all items where the usage price is higher than type $t$'s valuation and therefore does not factor into the revenue from type $t$, we instead map the probability mass to the trivial item. Under this mapping, type $t$'s utility in the modified lottery is the same as in the original one. Furthermore, the modified lottery is weakly worse for all other types because probability mass has been moved to the trivial item, which has zero value for all types. Hence the modified menu remains IC and yields the same buyer utilities and seller revenue.
\end{proof}


\subsection{Computing profit-maximizing menus.} \label{a:two-usage-prices-suffice} \label{a:upfront-price-formula-two-types} \label{a:partition-revenue} \label{a:value-proximity} \label{a:value-efficient}

\subsubsection{Two usage prices suffice} 

\begin{proof}[Proof of \cref{two-usage-prices-suffice}]
    
    Consider a IC menu of the form $\cC^t = (a^t, w^t, \mathbf{x}^t)$. The utility of contract $\cC^t$ for type $u$ is
    \begin{align*}
        U(u;\cC^t) &:= -w^t + \sum_{q} p^{a^t}_q \cd \max\bc{0, v^u_q - x^t_q}
        \\&= -w^t + \sum_{q: v^u_q \ge x^t_q} p^{a^t}_q \bp{v^u_q - x^t_q}.
    \end{align*}
    The revenue from type $t$ is
    \[w^t + \sum_{q\in S^t} p^{a^t}_q x^t_q,\] recalling that $S^t := \bc{q: v^t_q \ge x^t_q}$ is the set of outcomes that type $t$ accepts.
    Consider replacing $\cC^t$ with the contract
    \begin{align*}
        \cC'^t &= \bp{a'^t, {w'}^t, \mathbf{x'}^t} \\
        a'^t &= a^t \\
        w'^t &= w^t + \sum_{q\in S^t} p^{a^t}_q x^t_q \\
        \mathbf{x'}^t_q &= \case{0 & q\in S^t \\ \infty & \teoth.}
    \end{align*}
    We compute
    \begin{align*}
        U(u;\cC'^t) &= -w'^t + \sum_{q\in S^t} p^{a^t}_q v^u_q
        \\& = -w^t + \sum_{q\in S^t} p^{a^t}_q \bp{v^u_q - x^t_q}.
    \end{align*}
    Note that $U(t; \cC'^t) = U(t; \cC^t)$, so type $t$ prefers $\cC'^t$ to any other contract $\cC^u$ by IC of the original menu. The profit from $\cC'^t$ is exactly the profit from $\cC^t$ since the upfront price is increased by exactly the amount that the type $t$ usage prices are decreased, weighted by the outcome probabilities. The profit from all other types is unchanged because only $\cC^t$ is modified to get from the original menu to the modified menu.

    We finish by showing IC of the modified menu using IC of the original menu. It suffices to show that $U(u;\cC'^t) \le U(u; \cC^t)$ for all types $u\neq t$, which implies that type $u$ continues to choose $\cC^u$ in the modified menu. For each outcome $q$, we compare the corresponding term in the formulas for $$U(u;\cC'^t) = -w^t + \sum_{q\in S^t} p^{a^t}_q \bp{v^u_q - x^t_q}$$ and $$U(u;\cC^t) = -w^t + \sum_{q: v^u_q \ge x^t_q} p^{a^t}_q \bp{v^u_q - x^t_q}.$$ We split into two cases. If $v^u_q < x^t_q$ then the term for outcome $q$ in $U(u;\cC'^t)$ is negative while the term does not exist in $U(u;\cC^t)$. If $v^u_q \ge x^t_q$, the term corresponding to outcome $q$ in $U(u;\cC'^t)$ is either equal to the term in $U(u;\cC^t)$, which is nonnegative, if $q\in S^t$, or does not exist if $q\notin S^t$. In either case, the term for outcome $q$ contributes less to the sum in $U(u;\cC'^t)$ than $U(u;\cC^t)$, so $U(u;\cC'^t) \le U(u; \cC^t)$ as desired. 
\end{proof}

\subsubsection{Profit-maximizing upfront prices for two types} \label{sec:two-type-formula}

Consider a problem instance with two types and a single action. Suppose the two types are equally likely: $\mu_1 = \mu_2 = \frac{1}{2}$. The profit-maximizing menu consists of two contracts, with the single action $a$ used in both of them. Assume without loss of generality that type 2 is a highest-revenue type, so by \cref{highest-type-no-usage-prices} we can assume $x^2_q = 0$ for all $q$. Hence the menu search space consists of choosing $x^1_q\in \bc{0,\infty}$ for all $q$ as well as setting the upfront prices $w^1, w^2$. In the following lemma we derive explicit expressions for the optimal upfront prices fixing the usage prices $x^1$.

\begin{lemma} \label{upfront-price-formula-two-types}
Suppose $\mu_1 = \mu_2 = \frac{1}{2}$. Let type 2 be a highest-revenue type. In a menu which maximizes profit given a fixed usage price vector $x^1$, denote by $S \coloneq \bc{q:x^1_q = 0}$ the set of outcomes that type 1 accepts. Then
    \begin{align*}
        w^1 &= \sum_{q\in S} p^a_q v^1_q \\
        w^2 &= \min \bc{\sum_{q\notin S} p^a_q v^2_q + \sum_{q\in S} p^a_q v^1_q, \sum_q p^a_q v^2_q}.
    \end{align*}
\end{lemma}

\begin{proof}
    Note that increasing both upfront prices $w^1$ and $w^2$ at the same rate increases profit while preserving IC since each buyers' utility for each contract decreases at the same rate. This process stops when one buyer's surplus is 0. Hence any profit-maximizing menu must have no buyer surplus for at least one type. To finish the proof of the lemma, we split into two cases.

    \begin{enumerate}
        \item[] \tbf{Case 1.} Type 1 has no buyer surplus, so $$w^1 = \sum_{q\in S} p^a_q v^a_q.$$ Since type $2$ is a highest-revenue type, by \cref{highest-type-no-usage-prices}, type 2's usage prices are 0 in a profit-maximizing menu. The type 2 IC constraint is $$\sum_q p^a_q v^2_q - w^2 \ge \sum_{q\in S} p^a_q v^2_q - w^1 \iff w^2 \le w^1 + \sum_{q\notin S} p^a_q v^2_q$$ and the type 2 IR constraint is $$w^2 \le \sum_q p^a_q v^2_q.$$ Since the menu is assume to be profit-maximizing, one of these two conditions must bind, so $$w^2 = \min \bc{\sum_{q\notin S} p^a_q v^2_q + \sum_{q\in S} p^a_q v^1_q, \sum_q p^a_q v^2},$$ precisely what the lemma states.


        \item[] \tbf{Case 2.} Type 1 has positive buyer surplus, so the type with no buyer surplus is the higher type 2. By \cref{highest-type-no-usage-prices}, note that type 2's usage prices are 0 in a profit-maximizing menu. We can assume that the strict inequality $w^2 > w^1$ is true, otherwise if $w^2 = w^1$ we could have treated the type with no buyer surplus as the lower type 1 in which case \tbf{Case 1} applies.

        Since the menu is assumed to be profit-maximizing, increasing $w^1$ slightly to increase profit cannot be possible. Note that type 2 IC is not violated when $w^1$ increases since $\cC^1$ becomes \emph{less} attractive to type 2. Since type 1 has positive buyer surplus, type 1 IR is not violated by increasing $w^1$ slightly, so type 1 IC must bind, meaning type 1 prefers contract $\cC^2$ as much as $\cC^1$. However, if we now replace $\cC^1$ with $\cC^2$, type 1 revenue increases from $w^1$ to $w^2$, contradicting the assumption the menu is profit-maximizing. 
    \end{enumerate} 
\end{proof}

Recalling that we assumed $\mu^1 = \mu^2 = \fr12$, for a given subset $S$ of outcomes the seller revenue is

\begin{equation} \label{revenue-formula-two-types}
\fr12 \bp{w^1 + w^2} = \frac{1}{2} \left( \sum_{q \in S} p^a_q v^1_q + \min \bc{\sum_{q\notin S} p^a_q v^2_q + \sum_{q\in S} p^a_q v^1_q, \sum_q p^a_q v^2_q} \right),
\end{equation}

From \cref{revenue-formula-two-types}, we observe the following:

\begin{claim} \label{q-in-s}
     For outcomes $q$ such that $v^1_q \ge v^2_q$, seller profit is larger when $q \in S$ than when $q\notin S$.
\end{claim}

\begin{proof}
    Setting $x^1_q = 0$ for any $q$ satisfying $v^1_q \ge v^2_q$ contributes $p^a_q v^1_q$ to the sum $\sum_{q \in S} p_q^a v_q^1$ in \cref{revenue-formula-two-types} and $p^a_q v_q^1$ to the sum $$\sum_{q\notin S} p^a_q v^2_q + \sum_{q\in S} p^a_q v^1_q$$ in the minimum in \cref{revenue-formula-two-types}. On the other hand, setting $x^1_q = \infty$ contributes 0 and $p^a_q v^2_q$ to these respective sums. Since $v^1_q \ge v^2_q$, both contributions are weakly greater when $x^1_q = 0$, so it is optimal to set $x^1_q = 0$ implying $q\in S$. 
\end{proof}

\subsubsection{Reducing \tsf{Partition} to service provider problem for two types}

\begin{proof}[Proof of \cref{partition-revenue}]

First we claim that for both the \emph{if} and \emph{only if} directions we can assume that type 2 is a highest type. If type 1, whose revenue is at most $2M$, is a highest type, then the maximum seller profit is $2M < \fr{9M}{4}$. Hence if a menu's expected profit is $\fr{9M}{4}$ then type 1 cannot be a highest type. Also, if there exists a subset of $\bc{n_1,n_2,\lds,n_k}$ that sums to $\fr{M}{2}$ then we will show that there exists an menu in which type 2 is a highest type such that the seller profit is $\fr{9M}{4}$. 

\newcommand{\Szero}{S\sm \bc{0}}

As before, let $S$ denote the outcomes that type 1 accepts, which includes outcome 0, and define $M_S = \sum_{q\in \Szero} n_q$. By \cref{upfront-price-formula-two-types}, the profit-maximizing upfront prices are
\begin{align*}
    w^1 &= \sum_{q\in S} p^a_q v^1_q = M + M_S \\
    w^2 &=  \min\bc{\sum_{q\notin S} p_q v^2_q + w^1, \sum_q p_q v^2_q}
    = \min\bc{4M - 2M_S, 3M}.
\end{align*}
If $M_S > \fr{M}{2}$ then seller profit is $$ \frac12 w^1 + \frac12 w^2 = \fr12 \bp{(M + M_S) + (4M - 2M_S)} = \fr{5M - M_S}{2} < \fr{9M}{4},$$ and if $M_S < \fr{M}{2}$ then the seller profit is $$\fr12\bp{(M + M_S) + (3M)} = \fr{4M + M_S}{2} < \fr{9M}{4}.$$ Finally, if $M_S = \fr{M}{2}$ then seller profit is exactly $\fr{9M}{4}$ as $$w^1 = \fr{3M}{2} \le 3M = w^2.$$ We conclude that there is a subset $S$ with $M_S = \fr{M}{2}$ if and only if the maximum seller profit is $\fr{9M}{4}$. 
\end{proof}

\subsubsection{Continuity of the indirect profit objective}

The main technical ingredient in the proof of \cref{value-proximity} is the following result:
    

    \begin{lemma} \label{differ}
         For two states $s$ and $s'$ that differ in only one component $i = (t, t')$ by a constant $\eps>0$, we have $\Pi_{\tt{indirect}}(s') \ge \Pi_{\tt{indirect}}(s) - \eps$.

    \end{lemma}
    
    \begin{proof} Fix $\eps > 0$. Let $e_{t, t'} \in S$ denote the unit vector along the $(t, t')$-dimension.  In the following, we let $\mathbf{w}$ denote the optimal upfront prices in the definition of $\Pi_{\tt{indirect}}(s)$. Recall that $u_{\mathbf{w}}(s,t)$ denotes the optimal contract choice of a buyer of type $t$ at state $s$. We split our analysis into four cases.
        
\begin{enumerate}
    \item[] \textbf{Case 1.} \emph{Suppose $s' = s + \eps \cd e_{t, t'}$ for $t' = u_{\mathbf{w}}(s, t)$.} In this case, compared to the state $s$, the type $t$ buyer's value for the contract $t'$ has increased under the state $s'$. Since $t'$ was already the most preferred contract for the buyer under state $s$ and upfront prices $\mathbf{w}$, it is still remains so under the state $s'$ and the same upfront prices. Thus, under $\mathbf{w}$, the seller receives the same revenue under $s'$ and $s$. From this, it follows that $\Pi_{\tt{indirect}}(s') \geq \Pi_{\tt{indirect}}(s)$. 
    
    \item[] \textbf{Case 2.} \emph{Suppose $s' = s - \eps \cd e_{t, t'}$ for $t' \neq u_{\mathbf{w}}(s, t)$.} In this case, compared to the state $s$, the type $t$ buyer's value for the contract $t'$ has decreased under $s'$. Same as in the previous case, this implies that the buyer's most preferred contract remains the same in the two states under the same upfront prices, and hence once again, $\Pi_{\tt{indirect}}(s') \geq \Pi_{\tt{indirect}}(s)$.

    \item[] \textbf{Case 3.} \emph{Suppose $s' = s - \eps \cd e_{t, t'}$ for $t' = u_{\mathbf{w}}(s, t)$.} In this case, compared to the state $s$, the type $t$ buyer's value for the contract $t'$ has decreased under $s'$. Because of this, if the upfront prices $\mathbf{w}$ are unchanged, it is not guaranteed that the buyer continues to prefer contract $t'$ under $s'$. To ensure that the buyer's contract choice remains the same, one may try decreasing the upfront price $w^{t'}$ by $\eps$ (so that the buyer's utility for $t'$ remains the same). However, this reduction in the upfront price of contract $t'$ may change the optimal contract choice of other buyer types, potentially resulting in large drop in the seller's profit. 
    
    To avoid this issue, we choose carefully a set of contracts whose upfront prices are decreased by $\eps$.
    Let $D$ denote the set of all contracts $\tau \in [T]$ such that the seller's profit from contract $\tau$ is greater than that from contract $t' = u_{\mathbf{w}}(s,t)$:  
     $$D \coloneq \left\{ \tau \in [T] : w^{\tau} - c\bp{a^{\tau}} \geq  w^{t'} - c\bp{a^{t'}} \right\}.$$
    Note that $t' \in D$. Consider upfront prices $\mathbf{\bar{w}}$ defined as follows:
    \begin{align*}
        \bar{w}^\tau = \begin{cases}  \max(0, w^\tau - \eps)  & \teif \tau \in D \\
        w^\tau & \text{otherwise.}
        \end{cases}
    \end{align*}
    Consider the buyer of type $\tau \neq t$. If $u_{\mathbf{w}}(s, \tau) \in D$ and $w^\tau > \eps$ then we claim that $u_{\mathbf{\bar{w}}}(s', \tau) = u_{\mathbf{w}}(s, \tau)$, namely the type $\tau$ buyer prefers the same contract in state $s'$ under upfront prices $\mathbf{\bar{w}}$ as the one in state $s$ under $\mathbf{w}$. This is because, under $s'$ the contract values for type $\tau$ buyer remains the same, the upfront price of her most preferred contract reduces by $\eps$ and the upfront prices for other contracts reduces by at most $\eps$. We conclude that changing the upfront price to $\bar{w}^\tau$ reduces the service provider's profit from type $\tau \neq t$ by at most $\mu^\tau\cd \eps$, as either $w^\tau \le \eps$ and the profit was at most $\mu^\tau \cd \eps$ to begin with, or the type $\tau$'s contract choice remains the same and the upfront price reduces by $\eps$.

    On the other hand, if $u_{\mathbf{w}}(s, \tau) \notin D$, since only the upfront prices of contracts in $D$ are reduced, we conclude that either $u_{\mathbf{\bar{w}}}(s', \tau) = u_{\mathbf{w}}(s, \tau)$ or $u_{\mathbf{\bar{w}}}(s', \tau) \in D$. In the former case, the service provider's profit from type $\tau$ buyer remains the same. In the latter case, let $\bar{\tau} = u_{\mathbf{\bar{w}}}(s', \tau) \in D$. Then, the service provider's profit from type $\tau$ buyer in state $s'$ and upfront prices $\mathbf{\bar{w}}$ is $\bar{w}^{\bar{\tau}} - c(a^{\bar{\tau}}) = w^{\bar{\tau}}  - \eps - c(a^{\bar{t}}) \geq w^{t'} - c(a^{t'}) - \eps > w^{u_{\mathbf{w}}(s, \tau)} - c(a^{u_{\mathbf{w}}(s, \tau)}) - \eps$. Here, the first inequality follows from the definition of $D$ and from the fact that $\bar{\tau} \in D$, whereas the second inequality follows because $u_{\mathbf{w}}(s, \tau) \notin D$. Again, the service provider's profit from type $\tau \neq  t$ decreases by at most $\mu^\tau \cd \eps$.

    Next, consider the buyer of type $t$. Let $\bar{t} = u_{\mathbf{\bar{w}}}(s',t)$ denote the contract preferred by the buyer in state $s'$ under upfront prices $\mathbf{\bar{w}}$. We claim that $\bar{t} \in D$. To see why, observe first that for any contract not in $D$, the buyer's utility for the contract in state $s'$ and upfront prices $\mathbf{\bar{w}}$ is the same as that in state $s$ and upfront prices $\mathbf{w}$, since neither the contract value nor the upfront price has changed. Furthermore, since $t' = u_{\mathbf{w}}(s, t) \in D$, we have $s'_{t, t'} - \bar{w}^{t'}  = s_{t,t'} - w^{t'} $, and hence the buyer's utility for contract $t'$ also remains the same, since both contract value and upfront price have reduced by $\eps$. Finally, for other contracts in $D$, the buyer's utility is lower by $\eps$ as only the upfront price has reduced. Thus, we conclude the buyer will now prefer a contract $\bar{t}$ in $D$.  Thus, the service provider's profit from type $t$ buyer in state $s'$ and upfront prices $\mathbf{\bar{w}}$ is $\bar{w}^{\bar{t}} - c(a^{\bar{t}}) = w^{\bar{t}}  - \eps - c(a^{\bar{t}}) \geq w^{t'} - c(a^{t'}) - \eps$, where the  inequality follows from the definition of $D$. Thus, the service provider's profit from type $t$ buyer also reduces by at most $\mu(t) \eps$.
     
    Putting these cases together, we conclude that overall profit of the service provider in state $s'$ and under upfront prices $\mathbf{\bar{w}}$ is at most $\sum_{\tau} \mu^\tau \cd \eps = \eps$ less than that in state $s$ and under upfront prices $\mathbf{w}$. Thus, we obtain $\Pi_{\tt{indirect}}(s') \geq \Pi_{\tt{indirect}}(s) - \eps$.

    \item[] \textbf{Case 4.} \emph{Suppose $s' = s + \eps \cd e_{t, t'}$ for $t' \neq u_{\mathbf{w}}(s, t)$.}  This case is similar to the preceding case. Let $\hat{t} = u_{\mathbf{w}}(s,t)$, and note that $u_{\mathbf{w}}(s',t) \in \{\hat{t}, t'\}$. 
    
    Define $\hat{D}$ as the set of contracts that yield higher profit for the service provider than contract $\hat{t}$: $$\hat{D}  \coloneq \left\{ \tau \in [T] : w^{\tau} - c\bp{a^{\tau}} \geq  w^{\hat{t}} - c\bp{a^{\hat{t}}} \right\}.$$ Note that $\hat{t} \in \hat{D}$. If $t' \in \hat{D}$, then $u_{\mathbf{w}}(s', t) \in \hat{D}$, implying that the service provider's profit from type $t$ buyer has increased, and hence $\Pi_{\tt{indirect}}(s') \geq \Pi_{\tt{indirect}}(s)$. On the other hand, if $t' \notin \hat{D}$, then consider the upfront prices $\mathbf{\bar{w}}$ where $\bar{w}^{\tau} = w^\tau - \eps$ for $\tau \in \hat{D}$ and $\bar{w}^{\tau} = w^\tau$ otherwise. Using similar arguments as in the preceding case, it follows that in state $s'$ and under upfront prices $\mathbf{\bar{w}}$, each buyer types either prefer the same contract as in state $s$ and under upfront prices $\mathbf{w}$, or prefers a contract in $\hat{D}$. From this, it follows that the service provider's overall profit reduces by at most $\eps$, and hence $\Pi_{\tt{indirect}}(s') \geq \Pi_{\tt{indirect}}(s) - \eps$. 
\end{enumerate}
\end{proof}

We now use \cref{differ} to prove \cref{value-proximity}.

\begin{proof}[Proof of \cref{value-proximity}]
    Recall that we need to show that there exists an integer $G\ge 0$ such that if states $s$ and $s'$ are $r$-close for any $1 < r < 1 + \fr{1}{G+1}$ then $\Pi_{\tt{indirect}}(s') \ge r^{-G}\cd \Pi_{\tt{indirect}}(s)$. If $s$ and $s'$ are $r$-close then $s$ and $s'$ differ in all components by at most $(r-1) \cd \max_i s_i$. Applying \cref{differ} componentwise yields that $$\Pi_{\tt{indirect}}(s') \ge \Pi_{\tt{indirect}}(s) - T^2 \cd (r-1) \cd \max_i s_i$$ since $s$ has $T^2$ components. Note that $\Pi_{\tt{indirect}}(s)$ includes the maximum action cost $c$ for a profit-maximizing mapping of contracts to actions and hence is at least the seller revenue, which in turn is at least $\min_{t\in [T]} \mu^t \cd \max_i s_i$ since $s_i$ is the value of a contract for some buyer. To finish the proof, it suffices to choose $G$ so that $1-r^{-G} \ge \fr{1}{\min_{t\in [T]} \mu^t} \cd T^2\cd (r-1)$ so that using $\Pi_{\tt{indirect}}(s) \ge \min_{t\in [T]} \mu^t \cd \max_i s_i$ we have
    \begin{align*}
        \Pi_{\tt{indirect}}(s')
        &\ge \Pi_{\tt{indirect}}(s) - T^2 \cd (r-1) \cd \max_i s_i
        \\&\ge \bp{1 - \fr{1}{\min_{t\in [T]} \mu^t} \cd T^2\cd (r-1)} \cd \Pi_{\tt{indirect}}(s)
        \\&\ge r^{-G} \cd \Pi_{\tt{indirect}}(s)
    \end{align*}
    as desired. We compute
    \begin{align}
        & 1-r^{-G} \ge \fr{1}{\min_{t\in [T]} \mu^t} \cd T^2\cd (r-1) \notag
        \\&\iff  r^G - 1 \ge r^G \cd \fr{1}{\min_{t\in [T]} \mu^t} \cd T^2 \cd (r-1) \notag
        \\&\iff \sum_{j=1}^{G} r^{-j} \ge \fr{T^2}{\min_{t\in [T]} \mu^t}. \label{geometric-series} 
    \end{align}
    Finally, we choose $G = \fr{2T^2}{\min_{t\in [T]} \mu^t}$ so that for all $r>1$ with
    \begin{equation} \label{r-G-eq}
        r-1 \le \fr{1}{G+1} \imp \fr{G(G+1)}{2} \cd (r-1) \le \fr{G}{2} =  \fr{T^2}{\min_{t\in [T]} \mu^t},
    \end{equation}
    \cref{geometric-series} is true by
    \begin{align}
        \sum_{j=1}^{G} r^{-j} 
        &\ge \sum_{j=1}^{G} \bp{1 - j(r-1)} \label{bernou}
        \\&= G - \fr{G(G+1)}{2}\cd (r-1) \notag
        \\&\ge \fr{T^2}{\min_{t\in [T]} \mu^t}, \label{r-G-eq-app}
    \end{align}
    where \cref{bernou} follows from Bernoulli's Inequality $(1 + (r-1))^{-j} \ge 1 - j(r-1)$ and \cref{r-G-eq-app} follows from \cref{r-G-eq}. 
\end{proof}

\subsubsection{Seller profit for indirect menu is efficiently computable} 

\begin{proof}[Proof of \cref{value-efficient}.]
    The key idea to enumerate over all possible functions $u_\mathbf{w}(s, \cdot) :[T]\to \bc{0}\cup [T]$, of which there are $(T+1)^T$, a constant because we assume the number of types $T$ is a constant. Fixing the function $u_\mathbf{w}(s, \cdot) = u$, we claim that the service provider's optimization problem can be solved by linear programming. Indeed, fixing $u$ the profit is linear in the upfront prices $\bp{w^t}_{t\in [T]}$ and furthermore the $\argmax$ constraints added by fixing $u$ are precisely constraints analogous to the IC and IR constraints:
    \begin{align*}
        s_{t, u(t)} - w^{ u(t)} &\geq s_{t, t'} - w^{t'}, \quad \fl t,t' && \te{(IC_u)} \\
        s_{t, u(t)} - w^{ u(t)} &\geq 0, \quad \fl t  && \te{(IR_u)} 
    \end{align*}
Note that these constraints are linear in the upfront prices. We conclude that $\Pi_{\tt{indirect}}(s)$ can be computed by solving $(T+1)^T$ linear programs, one for each assignment mapping $u$ from types to contracts, and taking the maximum objective value over all feasible linear programs. 
\end{proof}

\subsection{The single-parameter setting} \label{a:no-single-contract-example} \label{a:higher-types-higher-values} \label{a:revenue-maximizing-upfront-prices} 

\subsubsection{Single contract cannot maximize seller profit when costs are heterogeneous} 

\begin{example}
    Consider the following single-parameter problem instance:
    \begin{itemize}
        \item Let $T=2$, $\mu = \bp{\fr23, \fr13}$, and $Q = \bc{1,2}$. There are two actions $a^1, a^2$ with costs $c(a^1) = 0, c(a^2) = \fr32$ such that action $a^i$ deterministically maps to outcome $i$.

        \item Valuations are given by $\mathbf{v}^1 = (1,2)$ and $\mathbf{v}^2 = (2,4)$.
    \end{itemize}
We can verify that the direct menu consisting of contracts $\cC^1 = (a^1, 1, \mbf{0}), \cC^2 = (a^2, 3, \mbf{0})$ is IC and yields seller profit $$\fr23 \cd \bp{w^1 - c(a^1)} + \fr13 \cd \bp{w^2 - c(a^2)} = \fr23 \cd (1-0) + \fr13 \cd \bp{3 - \fr32} = \fr76.$$ On the other hand, we claim that no menu consisting of a single contract $\cC = (a, w, \mbf{x})$ with $x_q \in \{0, \infty\}$ can achieve revenue 1. We split into two cases:
\begin{enumerate}
    \item[] \tbf{Case 1.} If $a = a^1$, then since $a^1$ deterministically maps to outcome 1, the usage price $x_2$ does not factor into the profit. To yield profit at all, we must have $x_1 = 0$. If $w \le 1$ then both types prefer $\cC$ to the trivial contract $\cC^0$, for a profit of at most 1. If $w\in (1, 2]$ then only type 2 prefers $\cC$ to $\cC^0$, for a profit of at most $\fr13 \cd 2 = \fr23$.

    \item[] \tbf{Case 2.} If $a = a^2$, then since $a^2$ deterministically maps to outcome 2, the usage price $x_1$ does not factor into the profit. To yield profit at all, we must have $x_2 = 0$. If $w \le 2$ then both types prefer $\cC$ to the trivial contract $\cC^0$, for a profit of at most $2 - \fr32 = \fr12$. If $w\in (2, 4]$ then only type 2 prefers $\cC$ to $\cC^0$, for a profit of at most $\fr13 \cd \bp{4 - \fr32} = \fr56$.
\end{enumerate}
We conclude that the maximum seller profit from a single-contract menu is at most 1 and hence a single contract cannot be profit-maximizing.

\end{example}

\subsubsection{Monotonicity of upfront prices and (scaled) contract values}

\begin{proof}[Proof of \cref{higher-types-higher-values}]
    Since $U(t;u) = \frac{\alpha^t}{\alpha^u} V^u - w^u$, the IC constraint requiring that a type $t$ buyer prefer contract $t$ over contract $u <  t$ yields
    \begin{align*}
         & V^t - w^t \ge \fr{\alpha^t}{\alpha^u} \cd V^u - w^u  \iff  w^t  - w^u \le \alpha^t \left( \fr{V^t}{\alpha^t} - \fr{V^u}{\alpha^u} \right).
    \end{align*}
    Similarly, the IC constraint requiring that a type $u$ buyer prefer contract $u$ over contract $t$ yields $w^t  - w^u \geq \alpha^u \left( \fr{V^t}{\alpha^t} - \fr{V^u}{\alpha^u} \right)$,
    and hence
    \begin{equation} \label{align:monotonic-payment}
        \alpha^u \left( \fr{V^t}{\alpha^t} - \fr{V^u}{\alpha^u} \right) \le w^t - w^u \le \alpha^t \left( \fr{V^t}{\alpha^t} - \fr{V^u}{\alpha^u} \right).
    \end{equation}
    Thus, we have
    \begin{align*}
        \bp{\alpha^t - \alpha^u} \left(\fr{V^t}{\alpha^t} - \fr{V^u}{\alpha^u}\right) \geq 0.
    \end{align*}
    Since $t > u$, we have $\alpha^t > \alpha^u$, and we conclude $\fr{V^t}{\alpha^t} \geq \fr{V^u}{\alpha^u}$ and also $w^t \geq w^u$ from \cref{align:monotonic-payment}.
\end{proof}

\subsubsection{Revenue-maximizing upfront prices and seller revenue for given contract values} 

\begin{proof}[Proof of \cref{revenue-maximizing-upfront-prices}]

Recall from \cref{align:monotonic-payment} that the IC constraints reduce to 
\begin{align*}
     \alpha^u \left( \fr{V^t}{\alpha^t} - \fr{V^u}{\alpha^u} \right) \le w^t - w^u \le \alpha^t \left( \fr{V^t}{\alpha^t} - \fr{V^u}{\alpha^u} \right) \quad \text{for all $t > u$.}
\end{align*}
Similarly, the IR conditions imply $w^t \leq V^t$ for all $t$. From \cref{higher-types-higher-values}, we have that $w^t \geq w^u$ and $\frac{V^t}{\alpha^t} \geq \frac{V^u}{\alpha^u}$ for $t > u$. Since $\alpha^t  > \alpha^u$, this in turn implies $V^t \geq V^u$ for $t  > u$. 

First, observe that, under the IC conditions, if the IR condition $w^1 \leq V^1$ holds, then it holds for all $t$. To see this, the IC condition for $t > 1$ implies
\begin{align*}
    w^t \leq w^1 + \alpha^t \left( \frac{V^t}{\alpha^t} - \frac{V^1}{\alpha^1} \right) \leq V^1  + \alpha^t \left( \frac{V^t}{\alpha^t} - \frac{V^1}{\alpha^1} \right) = V^t - \left(\frac{\alpha^t}{\alpha^1} - 1\right) V^1 \leq V^t.
\end{align*}
Thus, it suffices to impose the IC conditions along with $w^1 \leq V^1$ to ensure IR.

Next, suppose the IC condition holds for all $t$ and $u=t-1$. Then, for general $t> u$, we have
\begin{align*}
    w^t - w^u = \sum_{i=u}^{t-1} (w^{i+1} - w^i) \leq \sum_{i=u}^{t-1} \alpha^{i+1} \left( \frac{V^{i+1}}{\alpha^{i+1}} - \frac{V^i}{\alpha^i}\right) \leq \alpha^t \sum_{i=1}^{t-1} \left( \frac{V^{i+1}}{\alpha^{i+1}} - \frac{V^i}{\alpha^i}\right) = \alpha^t \left( \fr{V^t}{\alpha^t} - \fr{V^u}{\alpha^u} \right).
\end{align*}
Here, we have used the fact that $\alpha^{i+1} > \alpha^i$ in the second inequality. A similar argument shows that $w^t - w^u \geq \alpha^u \left( \fr{V^t}{\alpha^t} - \fr{V^u}{\alpha^u} \right)$. Thus, it is sufficient to impose the IC conditions for all $t, u$ with $u = t-1$. This yields the following set of inequalities:
\begin{align*}
     \alpha^{t-1} \left( \fr{V^t}{\alpha^t} - \fr{V^{t-1}}{\alpha^{t-1}} \right) \le w^t - w^{t-1} \le \alpha^t \left( \fr{V^t}{\alpha^t} - \fr{V^{t-1}}{\alpha^{t-1}} \right) \quad \text{for all $t>1$.}
\end{align*}
Along with $w^1 \leq V^1$, this implies that the seller's revenue is maximized if $w^1 = V^1$ and $w^t - w^{t-1} =\alpha^t \left( \fr{V^t}{\alpha^t} - \fr{V^{t-1}}{\alpha^{t-1}} \right)$ for all $t > 1$. Thus, the revenue-maximizing upfront prices are given by
\begin{align*}
    w^1 &= V^1\\
    w^t &= V^1 + \sum_{i=2}^{t} \alpha^i \left( \fr{V^i}{\alpha^i} - \fr{V^{i-1}}{\alpha^{i-1}} \right)  =  V^t  - \sum_{i=1}^{t-1} \left( \alpha^{i+1} - \alpha^i\right) \frac{V^i}{\alpha^i}.
\end{align*}
The expression for the seller's revenue follows after some simple algebra. 
\end{proof}

\end{document}